\documentclass{ectj}

\usepackage{amsfonts,amssymb,epsfig,verbatim,bm,latexsym,amsmath,url,amsbsy,graphicx,psfrag,enumerate,natbib,mdframed,float,multirow,rotating}

\newtheorem{theorem}{Theorem}
\newtheorem{assumption}{Assumption}

\newtheorem{lemma}{Lemma}
\newtheorem{example}{Example}
\newtheorem{remark}{Remark}
\newtheorem{algorithm}{Algorithm}

\renewcommand{\thesection}{\arabic{section}}
\renewcommand{\theequation}{\arabic{section}.\arabic{equation}}

\renewcommand{\thelemma}{\arabic{section}.\arabic{lemma}}

\newcounter{bean}
\newcounter{beana}

\DeclareMathOperator{\argmin}{argmin}

\newcommand{\On}{\mathcal{O}}
\newcommand{\E}{\mathbb{E}}
\newcommand{\indep}{\perp\!\!\!\perp}

\received{--}
\accepted{--}
\volume{0}
\title[Causal Inference in High-Dim GLM]{Causal Inference in High-Dimensional Generalized Linear Models with Binary Outcomes}
\author[J. Kong]{Jing~Kong$^{\dagger}$}
\address{$^{\dagger}$University of Southern California, University Park, Los Angeles, CA 90089, USA.}
\email{jingkong@usc.edu}

\begin{document}

\begin{abstract}
This paper proposes a debiased estimator for causal effects in high-dimensional generalized linear models with binary outcomes and general link functions. The estimator augments a regularized regression plug-in with weights computed from a convex optimization problem that approximately balances link-derivative–weighted covariates and controls variance; it does not rely on estimated propensity scores. Under standard conditions, the estimator is $\sqrt{n}$-consistent and asymptotically normal for dense linear contrasts and causal parameters. Simulation results show the superior performance of our approach in comparison to alternatives such as inverse propensity score estimators and double machine learning estimators in finite samples. In an application to the National Supported Work training data, our estimates and confidence intervals are close to the experimental benchmark.
\keywords{Causal Parameters, High-dimensional Models, Link Functions, Weighting}
\end{abstract}

\section{Introduction}\label{sec:intro}
\setcounter{equation}{0}
\setcounter{theorem}{0}
\setcounter{assumption}{0}
\setcounter{proposition}{0}
\setcounter{corollary}{0}
\setcounter{lemma}{0}
\setcounter{example}{0}
\setcounter{remark}{0}
Causal inference from observational data is central across economics, epidemiology, and related fields. Binary outcomes are common, and generalized linear models (GLMs) provide a convenient framework. This paper proposes a method to estimate causal parameters in high-dimensional GLMs with binary outcomes.

\subsection{Motivation}
In observational studies, an unconfoundedness assumption—treatment assignment is \emph{as good as random} conditional on covariates—is often invoked. To make this plausible, researchers include many covariates or flexible functions of them. When the number of covariates is large relative to the sample size, regularization method such as Lasso by \citet{tibshirani1996regression} is used, but shrinkage introduces bias and the sampling distribution of the original Lasso estimator is not generally tractable for inference.

For illustration, consider a single-covariate two-stage model:
\[Y_i \sim \text{Bernoulli}(g(\beta_{D}D_i+\beta_{X}X_i)), \quad D_i \sim \text{Bernoulli}(g^*(X_i)),\]
where $X_i \sim N(0,1)$ independently drawn across $i \in \{1,...,n\}$ for $n = 100$, and parameters $\beta_{D} = \beta_{X} = 1$, $g(x) = g^{*}(x) = \exp(x)/(1+\exp(x))$.
The respective regression plug-in estimates for a causal parameter: average treatment effect on the treated (ATET), using standard maximum likelihood estimates (MLE) and penalized MLE with Lasso penalty, are:
\[\hat{\tau}_{MLE} = \frac{1}{n_t}\sum_{\{i:D_i = 1\}} g(\hat{\beta}_{D,MLE}+\hat{\beta}_{X, MLE}X_i)-g(\hat{\beta}_{X, MLE}X_i),\]
\[\hat{\tau}_{Lasso} = \frac{1}{n_t}\sum_{\{i:D_i = 1\}} g(\hat{\beta}_{D,Lasso}+\hat{\beta}_{X, Lasso}X_i)-g(\hat{\beta}_{X, Lasso}X_i).\]
Figure \ref{fig:introduction} displays histograms of the estimated ATET using MLE and penalized MLE with Lasso each based on 500 replications; the dashed black line marks the true value. The two panels show that the MLE estimator is centred close to the truth and approximates a normal distribution, whereas the Lasso-based estimator exhibits bias and noticeable deviations from normality. This pattern is consistent with the fact that the regression plug-in with MLE performs well when the outcome model is correctly specified.
\begin{figure}
    \begin{center}
        \begin{minipage}[b]{0.45\textwidth}
        \includegraphics[scale=0.17]{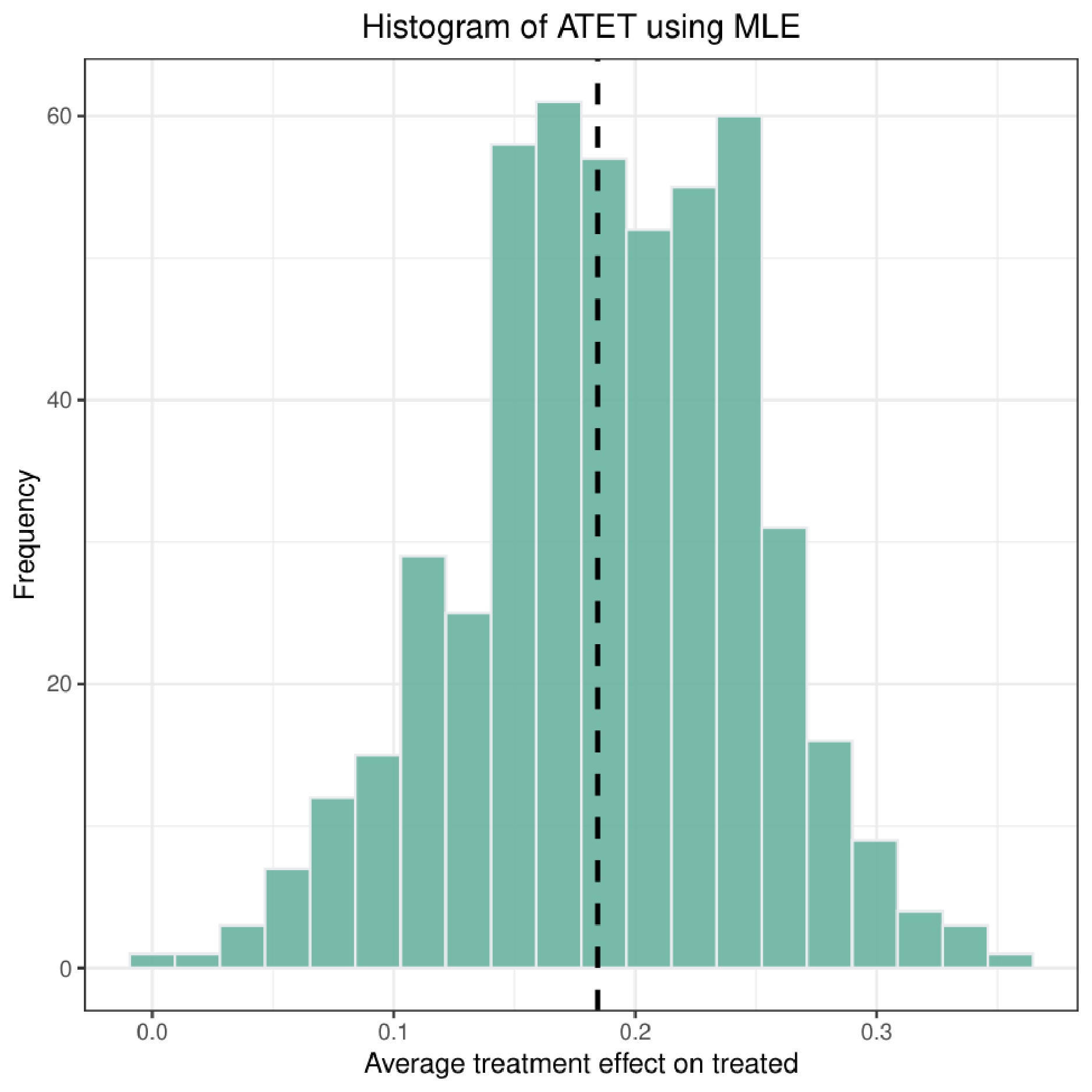}
    \end{minipage}
    \hspace{5mm}
    \begin{minipage}[b]{0.45\textwidth}
        \includegraphics[scale=0.17]{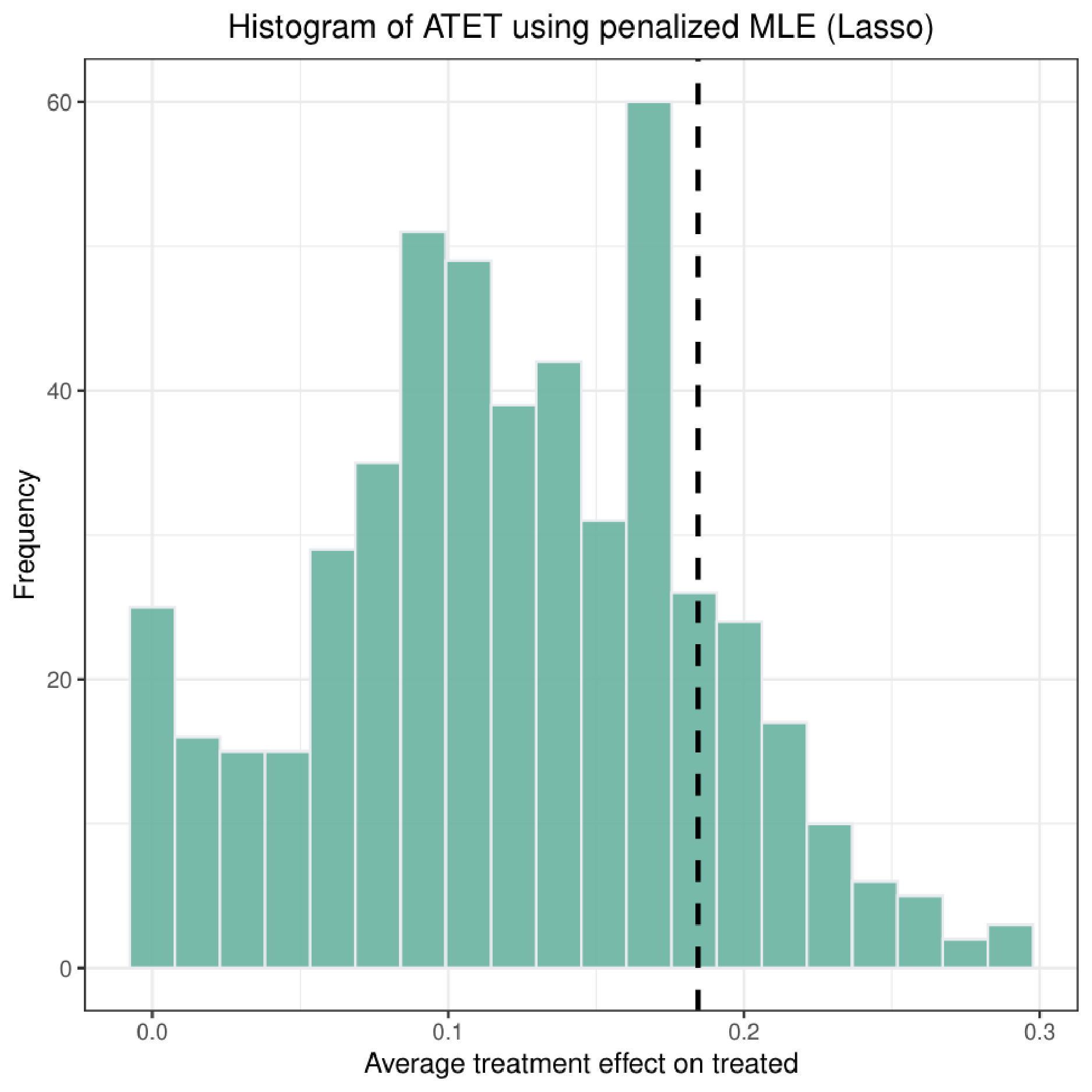}
    \end{minipage}
    \end{center}
        \caption{Histograms of the estimated ATET using MLE and Penalized MLE (Lasso)}
    \label{fig:introduction}
\end{figure}
This example underscores the need to mitigate bias for causal targets when regularization is necessary.

\subsection{Contributions and Related Work}
This paper proposes a debiased estimator for causal effects in high-dimensional generalized linear models (GLMs) with binary outcomes that does not rely on estimated propensity score. The estimator starts from a regularized GLM fit for the outcome model and then solves a single convex optimization to construct weights. The optimization program approximately balances covariate functions implied by the link derivative and simultaneously control the variance of the final estimator.

The key principle is a bias–variance trade-off. The GLM link function determines, through its derivative, which functions must be matched between treated and controls to remove the leading first-stage bias. The program therefore enforces approximate balance in these derivative-weighted covariates. At the same time, GLMs imply heteroskedastic residuals, so the objective includes an explicit penalty on the variance. In the special case of the identity link, corresponding to the linear model, the derivative is constant and the requirement reduces to equalizing treated and control covariate means, which coincides with standard covariate balance.

We establish two sets of asymptotic results: (i) for linear functionals $\xi^\top\beta_c$ (which may be dense), and (ii) for the causal parameter itself (a sample-average target rather than a fixed linear contrast). In both cases, the balance constraint makes the plug-in bias vanish faster than root-$n$, so the estimator is $\sqrt n$-consistent and asymptotically normal. We also provide a feasible variance estimator from the same optimization. 
In simulations (main text and supplement), the estimator performs better than inverse probability weighting (IPW), double/debiased machine learning (DML), automatic debiased machine learning (AML), and linear approximate residual balancing (ARB). The improvements are largest when the propensity model is complex. Confidence interval coverage approaches the nominal level as $n$ grows. The full-sample implementation has the lowest mean-squared error and the shortest intervals among our variants. In an empirical re-analysis of the National Supported Work (NSW) dataset, our estimate is closest to the experimental benchmark, and its 95\% confidence interval almost coincides with the randomized-experiment interval.

This paper contributes to three strands. First, it relates to doubly robust and orthogonal estimation in high dimensions. Regularized outcome regression can induce first-order bias in ATEs, as shown by \citet{belloni2014high}. Doubly robust and orthogonal estimators combine outcome and propensity models to obtain $\sqrt n$ inference under first-stage rate conditions, as developed by \citet{farrell2015robust}, \citet{belloni2017program}, and \citet{chernozhukov2018double}, and are related to debiased Lasso, as in \citet{van2014asymptotically} and \citet{zhang2014confidence}. We debias without inverting an estimated propensity score by imposing GLM-specific balance and minimizing variance in a single optimization problem.

Second, it connects to balancing-based debiasing and minimax linear estimation. Exact balancing via inverse-probability tilting or augmented scores, as in \citet{graham2012inverse} and \citet{graham2016efficient}, relies on inverse propensities and can be unstable in high dimensions. The ARB estimator of \citet{athey2018approximate} addresses linear models by balancing residualized covariates to remove regularization bias. We extend this idea to GLMs, where the relevant moments are link-derivative-weighted and the variance is heteroskedastic. Our optimization has a minimax-MSE interpretation and is related to augmented minimax linear estimation in \citet{hirshberg2021augmented}. Compared with ARB, we target GLM moments, use a different program, and obtain a feasible variance estimator. Riesz-based AML estimators of \citet{chernozhukov2022automatic} and \citet{chernozhukov2024automatic} avoid explicit propensities by learning the Riesz representer; our approach instead solves directly for weights that achieve debiasing in the GLM metric and explicitly optimizes variance, improving finite-sample stability.

Finally, the paper relates to inference in high-dimensional GLMs and to post-selection inference. \citet{cai2023statistical} provide debiased inference for single coefficients under general links using a two-step weighted correction. We allow dense linear contrasts and supply a single weighting scheme that controls both bias and variance. \citet{belloni2016post} and \citet{belloni2018uniformly} develop post-selection and uniform inference for linear models, including dense contrasts. Their procedures estimate both the outcome and an auxiliary propensity-type regression (double selection); our debiasing focuses on the outcome regression and an optimization that enforces balance.

\subsection{Organization and Notations}
The remainder of the paper is organized as follows. Section 2 formulates the problem and describes the estimator and algorithm. Section 3 develops the asymptotic theory for dense linear contrasts and causal parameters in a high-dimensional framework, deriving convergence rates and limit distributions. Section 4 reports simulation results. Section 5 presents an empirical application. Section 6 concludes. Proofs and additional simulation results are provided in the supplementary materials.

Throughout, for a vector $\boldsymbol{b} = (b_1,...,b_n)^\top \in \mathbb{R}^n$, we define the $\ell_p$ norm by $||\boldsymbol{b}||_p = \left(\sum_{i=1}^n|b_i|^p\right)^{1/p}$, and the $\ell_\infty$ norm by $||\boldsymbol{b}||_\infty = \max_{1\leq i \leq n} |b_i|$. For a matrix $A \in \mathbb{R}^{p\times q}$, $\sigma_i$ stands for the $i$-th largest singular value of $A$ and $\sigma_{\max}(A) = \sigma_1(A)$, $\sigma_{\min}(A) = \sigma_{\min\{p,q\}}(A)$.
For a smooth function $f(x)$ defined on $\mathbb{R}$, we define $f'(x) = df(x)/dx$ and $f''(x) = d^2f(x)/dx^2$. We define $\phi(x)$ and $\Phi(x)$ as the density function and cdf of the standard Gaussian random variable, respectively. We denote $\stackrel{p}{\rightarrow}$, $\stackrel{d}{\rightarrow}$ as convergence in probability and in distribution respectively.
For positive sequences $\{a_n\}$ and $\{b_n\}$, we write $a_n = o(b_n)$, $a_n \ll b_n$ or $b_n \gg a_n$ if $\lim_n a_n/b_n = 0$, and write $a_n = \On(b_n)$, $a_n \lesssim b_n$ or $b_n \gtrsim a_n$ if there exists a constant $C$ such that $a_n \leq Cb_n$ for all $n$. We write $a_n \simeq b_n$ if $a_n \lesssim b_n$ and $a_n \gtrsim b_n$.
For random variable sequences $\{X_n\}$ and positive sequences $\{a_n\}$, we write $X_n = o_p(a_n)$ if $X_n/a_n \stackrel{p}{\rightarrow} 0$, and write $X_n = O_p(a_n)$ if $\limsup_n P[|X_n/a_n| \geq M] \to 0$.

\section{Causal inference under high-dimensional GLM}
\setcounter{equation}{0}
\setcounter{theorem}{0}
\setcounter{assumption}{0}
\setcounter{proposition}{0}
\setcounter{corollary}{0}
\setcounter{lemma}{0}
\setcounter{example}{0}
\setcounter{remark}{0}
\subsection{Problem Formulation}
We estimate the average treatment effect on the treated (ATET) under the potential outcomes framework.\footnote{The average treatment effect (ATE) and the average treatment effect on the control (ATEC) can be handled analogously; they require modelling the missing potential outcome for the corresponding group.} For units $i=1,\ldots,n$, let $Y_i(1),Y_i(0)\in\{0,1\}$ denote potential outcomes, $D_i\in\{0,1\}$ the treatment indicator, $Y_i^{\text{obs}}=Y_i(D_i)$ the observed outcome, and $X_i\in\mathbb{R}^p$ a (possibly high-dimensional) covariate vector. We assume unconfoundedness for the control outcome and a GLM for $Y_i(0)\mid X_i$.

\begin{assumption}
    $D_i \indep Y_i(0)|X_i$ for all $i \in \{1,...,n\}$.
\end{assumption}

\begin{assumption}
    The response function conditional on covariates satisfies 
        \begin{equation}
            Y_i(0)|X_i \sim \text{Bernoulli}(g(X_i^\top \beta_c)),
        \end{equation}
    where $g: \mathbb{R} \to (0,1)$ is a known link function and $\beta_c \in \mathbb{R}^p$. In other words, we have $Y_i(0) = g(X_i^\top \beta_c)+\varepsilon_i$ where $\E(\varepsilon_i|X_i)=0$ and $var(\varepsilon_i|X_i)=g(X_i^\top \beta_c)(1-g(X_i^\top \beta_c))$ for all $i \in \{1,...,n\}$.
\end{assumption}

\begin{assumption}
    The general link function $g$ satisfies the following regularity conditions: (a) The link function $g$ is twice differentiable, monotonic increasing, Lipschitz on $\mathbb{R}$, and concave on $\mathbb{R}_{+}$; and for any $x\in \mathbb{R}$, it holds that $g(x)+g(-x)=1$; (b) There exist some constants $C_1$, $C_2>0$ such that, for all $x \geq 0$, $g(x) \leq \Phi(C_1x)$ where $\Phi(x)$ is the standard Gaussian cdf, and $\max\{g'(x)/(1-g(x)), x^2g'(x)\} < C_2$; (c) There exist some constants $c>0$ such that $\sup_{x\in \mathbb{R}}|g''(x)| \leq c$; (d) For $\ell_g(\beta)$ defined by
        \begin{equation}
            \ell_g(\beta) = -\frac{1}{n}\sum_{i=1}^n y_i\log\left[\frac{g(X_i^\top\beta)}{1-g(X_i^\top\beta)}\right]-\frac{1}{n}\sum_{i=1}^n\log(1-g(X_i^\top\beta)),
        \label{eq:lg}
        \end{equation}
        there exists some constant $C>1$ such that the Hessian matrix $\ell''_g(\beta)$ can be expressed as $\ell''_g(\beta) = \sum_{i=1}^n h(\beta;y_i,X_i)X_iX_i^\top/n$ for some $h(\beta;y_i,X_i)>0$ satisfying
        \begin{equation*}
            \max_{1\leq i\leq n}\left|\log h(\beta+b;y_i,X_i) - \log h(\beta;y_i,X_i)\right| \leq C(|X_i^\top \beta|^2+|X_i^\top b|^2+|X_i^\top b|).
        \end{equation*}
\end{assumption}
\textit{Assumption 2.3} allows many links used in practice, including but not limited to canonical links. The supplement verifies these conditions for the examples below.
\begin{example}
    \textnormal{The Logistic regression model prescribes the following link function
\[g(x) = \frac{\exp(x)}{1+\exp(x)},\]
which has a characteristic, $g'(x) = g(x)(1-g(x)) = \exp(x)/(1+\exp(x))^2$.}
\end{example}

\begin{example}
    \textnormal{Consider a latent variable model:
\[y_i^* = X_i^\top \beta + \epsilon_i, \quad y_i = \boldsymbol{I}(y_i^* \geq 0),\]
where $g(\cdot)$ is the cdf of $-\epsilon_i$. Beside the logistic link function, examples include:
\setcounter{bean}{0}
\begin{list}
{(\alph{bean})}{\usecounter{bean}}
    \item Probit link function: $g(x) = \Phi (x)$ where $\Phi(x)$ denoting the standard Gaussian cdf and $g'(x) = \phi(x)$ being the standard Gaussian pdf.
    \item Student's $t_v$-distributions with $\nu \in \mathbb{N} $: $g'(x) = \frac{\Gamma_d(\frac{\nu+1}{2})}{\sqrt{\nu \pi}\Gamma_d(\frac{\nu}{2})}\left(1+\frac{x^2}{\nu}\right)^{-(\nu+1)/2}$ where $\nu$ is the degree of freedom and $\Gamma_d$ is the gamma function.
\end{list}}
\end{example} 
\begin{remark}
    \textnormal{We treat the link as known (\textit{Assumption 2.2}).  When it is unknown, \citet{liu2021singleindex} propose an average-derivative estimator for fixed linear combinations of coefficients. Their estimator could in principle be combined with our balancing step; we leave this extension to future work.}
\end{remark}

Our target parameter is
\begin{equation}
    \tau := \frac{1}{n_t} \sum_{\{i:D_i = 1\}} \E[Y_i(1) - Y_i(0)|X_i],
\end{equation}
where $n_t$ is the number of treated units. Under \textit{Assumption 2.2},
\begin{equation}
    \tau = \mu_t - \mu_c,\quad \quad \mu_c = \frac{1}{n_t}\sum_{\{i:D_i = 1\}} g(X_i^\top \beta_c)
\end{equation}
We can estimate $\mu_t$ unbiasedly by a simple average of $\{Y_i^{\text{obs}}, D_i = 1\}$. The main challenge is high-dimensional estimation of $\mu_c$.

\subsection{Debiased Estimation Procedure}
We estimate $\mu_c$ by combining a regularized GLM plug-in with a residual reweighting step. The plug-in can suffer from regularization bias, especially when treated and control covariates differ and extrapolation is required. We therefore reweight control residuals using a convex program that (i) balances link-derivative–weighted covariates between treated and controls and (ii) controls variance.

\begin{algorithm}
\textnormal{
Input: design matrix $\boldsymbol{X} \in \mathbb{R}^{n \times p}$, treatment indicator $\boldsymbol{D} \in \{0,1\}^n$, outcomes $\boldsymbol{Y} \in \{0,1\}^n$.}
    \setcounter{bean}{0} 
\begin{center}
\begin{list}
{\textnormal{\textsc{Step} \arabic{bean}.}}{\usecounter{bean}}
\item \textnormal{Randomly split the control sample into two equal parts, $\mathcal{S}_1$ and $\mathcal{S}_2$, each containing $n_c$ observations.}
\item \textnormal{Fit a penalized GLM on $\mathcal{S}_1$:
        \begin{equation}
            \hat{\beta}_{c1} = \argmin_{\beta} \{\ell_g(\beta)+\lambda_n ||\beta||_1\}
        \end{equation}
        where $\ell_g(\beta)$ is defined by \eqref{eq:lg}.}\footnote{Throughout this paper, we use the lasso method to compute the plug-in regression estimate \(\hat{\beta}_{c1}\). However, other methods that offer similar \(L_1\) estimation error bounds can also be used.}
        
\item \textnormal{Compute weights $\gamma_i$ by solving
        \normalsize{
        \begin{equation*}
            \begin{aligned}
                \min_{\boldsymbol{\gamma}} &\left[\ \ \quad \quad (1-\zeta)\sum_{\{i \in \mathcal{S}_2\}} \gamma_i^2 g(X_i^\top \hat{\beta}_{c1})(1 - g(X_i^\top \hat{\beta}_{c1})) +\right. \\
                &\left. \quad \zeta \left\Vert \frac{1}{n_t}\sum_{\{i:D_i=1\}} g'(X_i^\top\hat \beta_{c1})X_i^\top -  \sum_{\{i \in \mathcal{S}_2\}} \gamma_i g'(X_i^\top\hat \beta_{c1})X_i^\top \right \Vert_\infty^2 \quad \right]
            \end{aligned}
        \end{equation*}
        }
        \normalsize{
        \begin{equation}
            s.t. \ \ \sum_{i \in \mathcal{S}_2} \gamma_i = 1, \quad 0\leq \gamma_i \leq \log(n_c)/n_c.
            \label{eq:optm_initial} 
        \end{equation}
        }}

\item \textnormal{Form
        \begin{equation}
            \hat{\mu}_c = \frac{1}{n_t} \sum_{\{i:D_i = 1\}}g(X_i^\top \hat{\beta}_{c1}) + \sum_{i \in \mathcal{S}_2} \gamma_i(Y_i - g(X_i^\top \hat{\beta}_{c1})).
        \label{eq:mu_chat}
        \end{equation}
        and estimate ATET by:
        \begin{equation}
            \hat{\tau} = \frac{1}{n_t}\sum_{\{i:D_i = 1\}} Y_i^{\text{obs}} - \hat{\mu}_c
            \label{eq:tau_hat}
        \end{equation}}
\end{list}
\end{center}
\end{algorithm}

Compared to the approach by \citet{athey2018approximate}, 
the key differences are: our objective accounts for GLM heteroskedasticity when penalizing variance, and the balance conditions use link-derivative–weighted covariates rather than residualized linear covariates. Relative to \citet{cai2023statistical}, who target single coefficients, our program targets $\mu_c$ which depends on all entries of $\beta_c$ through $g(X^\top\beta_c)$ and the treated covariates.

The optimization function in \eqref{eq:optm_initial} has a minimax interpretation. The objective combines a variance term and a squared worst-case design-conditional bias term; $\zeta\in[0,1]$ governs the trade-off. At $\zeta=1/2$ the optimization problem chooses weights that minimize a maximal mean squared error (MSE). This links our estimator to the augmented minimax linear estimator of \citet{hirshberg2021augmented}. Their analysis relies on Donsker classes; our high-dimensional setting does not, and we instead use standard assumptions for regularized GLMs.
    
\begin{remark}
    \textnormal{Note that the sample-splitting process is primarily used to simplify theoretical analysis and does not pose a limitation for the applications. 
    In simulations of Section 4, cross-fitting and full-sample implementations perform better by using all observations.
    This mirrors the strategy used by \citet{cai2023statistical}, who, while demonstrating the favourable properties of their estimator through sample-splitting, also effectively showed its practical efficacy using full sample in simulations.}
\end{remark}

\begin{remark}
    \textnormal{While our primary focus is on estimating the ATET, analogous estimators can be constructed, such as the ATE \(\sum_{i=1}^n \E[Y_i(1)-Y_i(0)|X_i]/n\), 
    and the ATEC \(\sum_{\{i:D_i=0\}} \E[Y_i(1)-Y_i(0)|X_i]/(n-n_t)\), by changing the target sample over which the missing potential outcomes are imputed.}
\end{remark}

\subsection{Comparison with Different Methods}
In this subsection, we focus on comparing different approaches for estimating $\mu_c$.\footnote{For simplicity, the intuitive comparison is based on no-splitting implementations for all estimators (direct plug-in and plug-in with different weightings).} First, let us consider a direct plug-in estimator, a plug-in with equal residual weights, and our plug-in plus optimized weights in~\eqref{eq:mu_chat}. The direct plug-in estimator is
\begin{equation}
    \hat{\mu}_{c,plug} = \frac{1}{n_t} \sum_{\{i:D_i=1\}} g(X_i^\top \hat \beta_c)
    \label{eq:mu_chat_plugin}
\end{equation}
and the plug-in estimator with simple averaging of residuals: 
\begin{equation}
    \hat{\mu}_{c,reg} = \frac{1}{n_t} \sum_{\{i:D_i=1\}} g(X_i^\top \hat \beta_c) + \frac{1}{n_c} \sum_{\{i: D_i=0\}} (Y_i^{obs}-g(X_i^\top \hat \beta_c))
    \label{eq:mu_chat_reg}
\end{equation}
where $\hat{\beta}_c$ is obtained by Lasso. 
Examine error analyses by contrasting the discrepancies among estimators $\hat{\mu}_{c,plug}$, $\hat{\mu}_{c,reg}$, $\hat{\mu}_{c}$ and the true value $\mu_c$:
\begin{equation}
    \hat \mu_{c,plug} -\mu_c = \left[\frac{1}{n_t}\sum_{\{i:D_i=1\}} g'(X_i^\top\hat \beta_c)X_i^\top \right](\hat{\beta}_c - \beta_c)
    + \frac{1}{n_t}\sum_{\{i:D_i=1\}} \Delta_i^t,
\label{eq:mu_chat_plugin_error}
\end{equation}
\begin{equation}
    \begin{aligned}
    \hat \mu_{c,reg} -\mu_c &= \left[\frac{1}{n_t}\sum_{\{i:D_i=1\}} g'(X_i^\top\hat \beta_c)X_i^\top -  \frac{1}{n_c}\sum_{\{i:D_i=0\}} g'(X_i^\top\hat \beta_c)X_i^\top \right](\hat{\beta}_c - \beta_c) \\
    &+ \frac{1}{n_c} \boldsymbol{1}_{n_c}^\top\boldsymbol{\varepsilon} + \frac{1}{n_t}\sum_{\{i:D_i=1\}} \Delta_i^t + \frac{1}{n_c}\sum_{\{i:D_i=0\}} \Delta_i^c,
    \label{eq:mu_chat_reg_error}
    \end{aligned}
\end{equation}
\begin{equation}
    \begin{aligned}
    \hat \mu_{c} -\mu_c &= \left[\frac{1}{n_t}\sum_{\{i:D_i=1\}} g'(X_i^\top\hat \beta_c)X_i^\top -  \sum_{\{i:D_i=0\}} \gamma_i g'(X_i^\top\hat \beta_c)X_i^\top \right](\hat{\beta}_c - \beta_c) \\
    &+ \boldsymbol{\gamma}^\top\boldsymbol{\varepsilon} + \frac{1}{n_t}\sum_{\{i:D_i=1\}} \Delta_i^t + \sum_{\{i:D_i=0\}} \gamma_i \Delta_i^c.
    \label{eq:mu_chat_error}
    \end{aligned}
\end{equation}
where $\Delta_i^{(\cdot)} = g''(X_i^\top\hat \beta_c + t_{i}^{(\cdot)}X_i^\top(\beta-\hat \beta_c))[X_i^\top(\hat \beta_c-\beta)]^2$ for some $t_{i}^{(\cdot)} \in (0,1)$.

In low dimensions, the direct plug-in estimator with MLE is $\sqrt{n_c}$-consistent and asymptotically normal. In high dimensions, plugging in regularized estimates need not retain these properties due to shrinkage bias. When treated and control groups are (approximately) balanced, the regression–imputation estimator in \eqref{eq:mu_chat_reg_error} improves on the plug-in \eqref{eq:mu_chat_plugin_error} by averaging residuals and can be valid under additional conditions. However, with noticeable imbalance, the simple residual average does not control the first linear term in \eqref{eq:mu_chat_reg_error}. In contrast, our estimator \eqref{eq:mu_chat_error} are flexible to be able to chooses weights to balance the derivative-weighted covariates and let the first term to be as small as possible.

By applying H\"older's inequality, the absolute value of \eqref{eq:mu_chat_error} can be bounded as follows:
\begin{equation}
    \begin{aligned}
        |\hat \mu_{c} -\mu_c| & \leq 
        \underbrace{\left\Vert \frac{1}{n_t}\sum_{\{i:D_i=1\}} g'(X_i^\top\hat \beta_c)X_i^\top -  \sum_{\{i:D_i=0\}} \gamma_i g'(X_i^\top\hat \beta_c)X_i^\top \right \Vert_\infty ||\hat{\beta}_c - \beta_c||_1}_{\text{(a) Main term of bias}} \\
        &+ \underbrace{\left|\sum_{\{i:D_i=0\}} \gamma_i \varepsilon_i \right|}_{\text{(b) Main term of variance}} + \underbrace{\left|\frac{1}{n_t}\sum_{\{i:D_i=1\}} \Delta_i^t \right| + \left|\sum_{\{i:D_i=0\}} \gamma_i \Delta_i^c\right|}_{\text{(c) Negligible terms}}.
    \label{eq:mu_chat_error_byH}
    \end{aligned}
\end{equation}
In Equation \eqref{eq:mu_chat_error_byH}, term (a) is the primary bias, which is our main concern, term (b) denotes the variance, anticipated to be of order $O_p(1/\sqrt{n}_c)$, and term (c) consists of negligible components.
The bias term (a) equals the imbalance in derivative-weighted covariates times the $L_1$ error of the first stage estimates.
Under standard conditions, this error scales as $O_p(k\sqrt{\log(p)/n_c})$. Therefore, our optimization should drive the imbalance to $O_p(\sqrt{\log(p)/n_c})$. 
to ensure that the entire bias term (a) is negligible relative to the variance term (b) when $k \ll \sqrt{n_c}/\log p$. 
This key insight is the main foundation of our proofs of theorems presented in Section 3.

We next consider benchmark our estimator against methods that use inverse propensity weighting. For example,
a pure weighting method $\hat{\mu}_{c,weight} = \sum_{\{i:D_i  =0\}} \gamma_i Y_i^{\text{obs}}$ with
\begin{equation}
    \gamma_i = \frac{\frac{p(X_i)}{1-p(X_i)}}{\sum_{\{i:D_i = 0\}}\frac{p(X_i)}{1-p(X_i)}}.
    \label{eq:mu_chat_weig}
\end{equation}
where $p(x) := \mathbb{P}(D=1|X=x)$ is the propensity score.
With nonparametric $p(\cdot)$ and a fixed number of covariates, this estimator is asymptotically linear and attains the semiparametric efficiency bound, as shown by \citet{hirano2003efficient}.
We also considers the Double Machine Learning (DML) estimator of \citet{chernozhukov2018double}, which combines regression imputation with inverse propensity weighting,
\begin{equation}
    \hat{\mu}_{c,DL} = \frac{1}{n_t} \sum_{\{i:D_i = 1\}} g(X_i^\top \hat{\beta}_c) + \sum_{\{i:D_i = 0\}} \frac{\frac{\hat{p}(X_i)}{1-\hat{p}(X_i)}(Y_i^{\text{obs}}-g(X_i^\top \hat{\beta}_c))}{\sum_{\{i:D_i=0\}}\frac{\hat{p}(X_i)}{1-\hat{p}(X_i)}}
    \label{eq:mu_chat_DL}
\end{equation}
With sample-splitting or cross-fitting, DML exhibits desirable asymptotic performance, as long as the estimations of both the outcome and propensity score model are moderately well.
However, in high dimensions and finite sample, inverse weights can be unstable. Specifically, a minor perturbation in the estimated value of $1-p(X_i)$ can translate into a significant effect, especially when it is close to 0, due to the involvement of its inverse in the weighting. This challenge is indicative of a broader issue that often emerges when inverting estimated values.
Our estimator differs in how the weights are formed: we strategically relies on the GLM structure to choose weights through a single convex program that balances link-derivative–weighted covariates and penalizes variance. This avoids the instability associated with inverse weights. Conversely, while our method is confined to the GLM setting, inverse propensity score weights debiasing accommodates any general form of outcome models, highlighting a flexibility not inherent to our approach.

Finally, we consider AML, as developed by \citet{chernozhukov2022automatic}. It avoids inverse propensity weights by estimating a Riesz representer $\alpha(\cdot)$ and then forming
\begin{equation}
    \hat{\mu}_{c,AL} = \frac{1}{n_t} \sum_{\{i:D_i = 1\}} g(X_i^\top \hat{\beta}_c) + \sum_{\{i:D_i = 0\}} \hat{\alpha}(X_i)(Y_i^{\text{obs}}-g(X_i^\top \hat{\beta}_c))
    \label{eq:mu_chat_AML}
\end{equation}
where \(\hat{\alpha}(\cdot)\) can be learned by a minimum-distance program (e.g., Lasso). 
In contrast to our method, which controls both bias and variance, the AML approach mainly focuses on bias reduction, even under heteroskedasticity case, through its goal of Riesz representer estimation.

\section{Theoretical Properties}
\setcounter{equation}{0}
\setcounter{theorem}{0}
\setcounter{assumption}{0}
\setcounter{proposition}{0}
\setcounter{corollary}{0}
\setcounter{lemma}{0}
\setcounter{example}{0}
\setcounter{remark}{0}
\subsection{Debiased Estimation for Dense Contrasts}
In the preceding section, the focus was on the causal parameter estimation. To formally establish its asymptotic properties, it is convenient to begin with a linear contrast $\xi^\top \beta_c$ where $\xi$ is a $p \times 1$ vector.
For instance, if $\xi$ is a basis vector like $(1,0,...,0)$, i.e., our target of interest is the first element of $\beta_c$, the problem becomes akin to the widely-studied debiased lasso estimator such as \citet{van2014asymptotically} and \citet{javanmard2014confidence}.
Here we allow $\xi$ to be dense and consider $\theta = \xi^\top \beta_c$. Linked to Algorithm 1, the estimator is given by
\begin{equation}
    \hat{\theta} = \xi^\top \hat{\beta}_{c1} + \sum_{i \in \mathcal{S}_2} \gamma_i(Y_i^{obs} - g(X_i^\top \hat{\beta}_{c1})), 
    \label{eq:est_dense}
\end{equation}
where $\boldsymbol{\gamma}$ solves the constrained program
\begin{equation*}
    \min_{\boldsymbol{\gamma}} \sum_{\{i\in \mathcal{S}_2\}} \gamma_i^2 g(X_i^\top \hat \beta_{c1})(1-g(X_i^\top \hat \beta_{c1}))
\end{equation*}
\begin{equation}
    \begin{aligned}
        s.t.\quad \quad \left\Vert \xi -  \sum_{\{i \in \mathcal{S}_2\}} \gamma_i g'(X_i^\top\hat \beta_{c1})X_i \right \Vert_\infty & \leq C_3 \sqrt{\log p/n}, \\
        \max_{i \in \mathcal{S}_2} |\gamma_i| & \leq C_4 \frac{\sqrt{\log n_c}}{n_c}.
        \label{eq:optm_dense}
    \end{aligned}
\end{equation}
Note that the optimization problem here is represented in a constraint form rather than the Lagrange form in Algorithm 1, facilitating theoretical proofs.
Intuitively, the goal of this problem is to find the weights that minimize the variance, while simultaneously ensuring that the bias remains constrained to a low level.
This formulation bridges important cases. For $\xi=e_1$ it aligns with \citet{cai2023statistical}. For $\xi=\frac{1}{n_t}\sum_{D_i=1} g'(X_i^\top\hat\beta_{c1})X_i$, $\hat\theta$ is equivalent to $\hat\mu_c$. 

It should be emphasized that applying the method of \citet{cai2023statistical} directly within our framework is not suitable. 
Specifically, Theorem 1 in their work indicates that their debiased estimator \(\tilde{\beta}\) satisfies \(\sqrt{n_c}(\tilde{\beta}_{c,i} - \beta_{c,i}) = A_{n,i} + B_{n,i}\), where \(A_{n_i}\) is asymptotically normally distributed and \(B_{n_i}\) is negligible per coordinate. But the aggregated error term \(\sum_{i=1}^p \xi_i B_{n,i}\) need not be negligible, as it could increase polynomially with the dimension of covariates $p$.
In addition, the simultaneous inference approach suggested by Section 2.3 of \citet{cai2023statistical} is conservative for our purposes, given our interest in a linear combination of coefficients rather than a subset of coefficients. 
Therefore, our approach extends \citet{cai2023statistical}, accommodating various \(\xi\) values and thereby generalizing from the debiasing of individual coefficients to encompass any combination of coefficients.

We impose two additional assumptions:
\begin{assumption}
    $\{X_i\}_{1\leq i\leq 2n_c}$  satisfies $X_i = Q_i \Sigma^{1/2}$, where $\E(Q_{ij}) = 0$, var$(Q_{ij})=1$, for all $i$ and $j$, and all individual entries $Q_{ij}$ are sub-Gaussian and independent.
\end{assumption}
\begin{assumption}
    Let that $\Sigma = \E[X_iX_i^\top] \in \mathbb{R}^{p\times p}$ and 
    \[\Theta(k) = \{(\beta_c,\Sigma):||\beta||_0 \leq k, ||\beta||_2 \leq M_1, M_2^{-1} \leq \lambda_{\min}(\Sigma)\leq \lambda_{\max}(\Sigma) \leq M_2\}\]
    for constants $M_1>0$ and $M_2>1$ independent of $n_c$ and $p$. Suppose $(\beta_c, \Sigma) \in \Theta(k)$.
\end{assumption}

While \textit{Assumption 3.1} may impose restrictions in practical applications, it is an important technical condition to ensure the feasibility of the minimization problem.
In causal estimation, this assumption can be relaxed under the overlap assumption. 
\textit{Assumption 3.2}, on the other hand, is conventional within the literature under the exact sparsity and GLM framework.\footnote{The exact sparsity can be weakened to approximate sparsity, see \citet{belloni2016post}, as long as the conditions on the estimation error bound hold, since our theoretical results depend only on the error bound on coefficients estimation. For simplicity, we restrict our attention to the exact sparsity case for theoretical properties. However, in our simulation study (Section 4), we include a decaying-coefficient design, illustrating performance when the approximate sparsity holds.}
The positive definiteness of the covariance matrix $\Sigma$ ensures the restricted strong convexity of the loss function, playing a pivotal role in delineating the Lasso risk bounds, such as in \citet{negahban2012unified} and \citet{cai2023statistical}.

\begin{lemma}
    Assuming \textit{Assumptions 2.2-2.3} and \textit{3.1-3.2} hold and given that $\xi^\top \Gamma^{-2} \xi \leq V$ for some constant $V>0$, where
    \begin{equation*}
        \Gamma = \E [g'(X_i^\top\hat \beta_{c1})X_i X_i^\top],
    \end{equation*} 
    it follows that the optimization problem defined by \eqref{eq:optm_dense} is feasible with probability tending to 1.
    Specifically, the constraints are met by setting $\gamma_i^* = \frac{1}{n_c}\xi^\top \Gamma^{-1}X_i$.
\label{lemma:optm}
\end{lemma}

Lemma \ref{lemma:optm} plays a pivotal role in facilitating subsequent theorems related to the asymptotic analysis. 
It provides a feasible solution, $\gamma_i^*$, that forms an upper bound for the asymptotic variance of the debiased estimator. 
Given that $\Gamma$ is invertible for any given $\hat{\beta}_{c1}$ by \textit{Assumptions 2.3} and \textit{3.2}, the condition $\xi^\top \Gamma^{-2} \xi \leq V$ fundamentally serves to constrain the magnitude of $||\xi||_2$.

\begin{theorem}
    Assuming the conditions of Lemma \ref{lemma:optm} and the following constraints: a minimum estimand size of 
    $||\xi||_{\infty} \geq \kappa >0$, 
    a lasso penalty parameter is $\lambda_n = O(\sqrt{\frac{\log p}{n_c}})$, and a sparsity level $k \ll \frac{\sqrt{n_c}}{\log p \log n_c}$,
    the estimator defined by \eqref{eq:est_dense} is asymptotically normal. Specifically,
    \[(\hat{\theta} - \theta)/\sqrt{v} \stackrel{d}{\rightarrow} N(0,1),\]
    with $v:=\sum_{i\in \mathcal{S}_2} \gamma_i^2 g(X_i^\top \beta_c)(1-g(X_i^\top \beta_c))$ and $v = O_p(1/n_c)$. In addition, the variance $v$ can be consistently estimated by $\hat{v}:=\sum_{i \in \mathcal{S}_2} \gamma_i^2 g(X_i^\top\hat{\beta}_{c1})(1-g(X_i^\top\hat{\beta}_{c1}))$, i.e., $\hat{v} \stackrel{p}{\rightarrow} v$.
\label{thm:dense}
\end{theorem}

The lower bound on $\|\xi\|_\infty$ excludes near-zero targets such as $\xi = (1/\sqrt{p},...,1/\sqrt{p})$.
The selection of the lasso tuning parameter's order aligns with established standards found in the literature, as exemplified by \citet{bickel2009simultaneous}. 
However, it is important to note that the sparsity level requirement is slightly stronger than that in \citet{athey2018approximate}, where $k \ll \frac{\sqrt{n_c}}{\log p}$. 
This requirement in the condition is necessitated by the GLM framework; the error decomposition \eqref{eq:mu_chat_error_byH} introduces Taylor expansion's remainder terms, and a stronger constraint is required to ensure their negligibility.

\subsection{Debiased Estimation for Causal Parameters}
We now analyse ATET.
The main distinction lies in the vector $\xi$: it is now considered as a random vector, whereas in the previous analysis, it was treated as a deterministic vector. 
More precisely, we have $\xi = \frac{1}{n_t}\sum_{{i:D_i=1}} g'(X_i^\top \hat{\beta}_{c1})X_i$, leading to:
\begin{equation}
    \hat{\mu}_c - \mu_c = \hat{\theta} - \theta +\frac{1}{n_t} \sum_{i:D_i = 1} \Delta_i^t.
\end{equation}
The second term is the second-order remainder of the Taylor expansion and negligible under certain conditions.
For this target we replace \textit{Assumption 3.1} with conditions more common in causal inference:
\begin{assumption}
    $\{X_i\}_{1\leq i\leq n}$ conditional on $D_i=d$, where $d\in\{0,1\}$, are independent and identically distributed sub-Gaussian random vectors, that is, there exists a constant $c \in \mathbb{R}$ satisfying $\E[\exp\{v^\top X\}|D=d]\leq \exp\{||v||_2^2 c^2/2\}$ for all $v \in \mathbb{R}^p$ after re-centring.
\end{assumption}
\begin{assumption}
    There is a constant $\epsilon>0$ such that $\epsilon \leq p(x)\leq 1-\epsilon$ for all $x\in \mathbb{R}$.
\end{assumption}

\begin{theorem}
    Under \textit{Assumptions 2.1-2.3} and \textit{3.2-3.4} and assuming that the overall odds of receiving treatment, denoted by $\mathbb{P}(D=1)/\mathbb{P}(D=0)$, converge to a limit $ \rho \in (0,\infty)$, 
    along with the constraints on the sparsity level $k \ll \frac{\sqrt{n_c}}{\log p \log n_c}$, the lasso penalty parameter $\lambda_n = O(\sqrt{\log p/n_c})$ and $\ell_\infty$ norm of $||\E[g'(X^\top \hat{\beta}_{c1})X|D=1]||_\infty \geq \kappa>0$, 
    the estimator defined by \eqref{eq:mu_chat}, with weights as given by the constrained form of \eqref{eq:optm_initial}, is asymptotically normal:
        \[(\hat{\mu}_c - \mu_c)/\sqrt{v_c} \stackrel{d}{\rightarrow} N(0,1),\]
    where $v_c :=\sum_{i\in \mathcal{S}_2} \gamma_i^2 g(X_i^\top \beta_c)(1-g(X_i^\top \beta_c))$, $v_c= O_p(1/n_c)$ and the variance estimate $\hat{v}_c:=\sum_{i \in \mathcal{S}_2} \gamma_i^2 g(X_i^\top\hat{\beta}_{c1})(1-g(X_i^\top\hat{\beta}_{c1}))$ converges in probability to $v_c$ at a rate faster than $1/n_c$.
    Consequently, the following results hold:
    \begin{list}{(\alph{bean})}{\usecounter{bean}}
        \item For the $\mu_c$ estimation:
        \[(\hat{\mu}_c - \mu_c)/\sqrt{\hat{v}_c} \stackrel{d}{\rightarrow} N(0,1).\]
        \item For the causal parameter estimation, given by $\hat{\tau} = \bar{Y}_t - \hat{\mu}_c$, the following result holds:
        \begin{equation}
            \frac{\hat{\tau} - \tau}{\sqrt{\hat{v}_\tau}} \stackrel{d}{\rightarrow} N(0,1)
        \end{equation}
        with $\hat{v}_\tau = \hat{v}_c+\hat{v}_t$ and $\hat{v}_t := \frac{1}{n_t^2}\sum_{\{i:D_i = 1\}}(Y_i - \bar{Y}_t)^2$, where $\hat{v}_\tau = O_p(1/n)$.
    \end{list} 
    \label{thm:formalCF}
\end{theorem}
Theorem \ref{thm:formalCF} establishes root-$n$ inference for the treatment effect without requiring the consistent estimation of the propensity score, and provides feasible standard errors.

\section{Simulation Study}
\setcounter{equation}{0}
\setcounter{theorem}{0}
\setcounter{assumption}{0}
\setcounter{proposition}{0}
\setcounter{corollary}{0}
\setcounter{lemma}{0}
\setcounter{example}{0}
\setcounter{remark}{0}
We study a two-stage model with $Y\sim\mathrm{Bernoulli}\{g(X^\top\beta_Y+D)\}$ and $D\sim\mathrm{Bernoulli}\{g(X^\top\beta_D)\}$, where $g(x)=\exp(x)/(1+\exp(x))$. Covariates are $X\sim N(0,\Sigma)$ with $\Sigma_{ij}=\rho^{|i-j|}$. 
The outcome coefficients are sparse, $(\beta_Y)_j\propto 1/j^2$ and the propensity coefficients are either sparse $(\beta_D)_j\propto 1/j^2$ or dense $(\beta_D)_j\propto 1/\sqrt{j}$. 
We compare our estimator with sample splitting (DB1), cross-fitting (DB2), and no splitting (DB6) to Na\"{i}ve, Regression with simple averaged residuals by \eqref{eq:mu_chat_reg} , IPW by \eqref{eq:mu_chat_weig}, DML by \eqref{eq:mu_chat_DL}, AML by \eqref{eq:mu_chat_AML}, and ARB. Performance is summarized by relative mean-squared error (MSE).\footnote{More designs, full tuning details, and additional metrics are reported in the Online Supplement.}

\begin{table}
\label{tab:RMSE_1}
  \caption{Mean‐squared error of different estimators}
  \begin{center}
    \begin{tabular*}{\textwidth}{@{}lcccccccc@{}}
      \hline\hline
      \multirow{3}{*}{Method}
        & \multicolumn{4}{c}{Sparse $\beta_D$}
        & \multicolumn{4}{c}{Dense $\beta_D$} \\
      \cline{2-5} \cline{6-9}
        & \multicolumn{2}{c}{$\|\beta_D\|_2=1$}
        & \multicolumn{2}{c}{$4$}
        & \multicolumn{2}{c}{$\|\beta_D\|_2 = 1$}
        & \multicolumn{2}{c}{$4$} \\
      \cline{2-3}\cline{4-5}\cline{6-7}\cline{8-9}
        & $\|\beta_Y\|_2=1$ & $4$
        & $1$ & $4$
        & $\|\beta_Y\|_2=1$ & $4$
        & $1$ & $4$ \\
      \hline
      Naïve      & 1.176 & 4.565 & 2.214 & 8.785 & 0.430 & 1.539 & 0.646 & 2.338 \\
      Regression & 0.162 & 0.224 & 0.262 & 0.411 & 0.112 & 0.165 & 0.144 & 0.210 \\
      IPW        & 0.196 & 0.630 & 0.373 & 1.222 & 0.162 & 0.487 & 0.227 & 0.711 \\
      DML1       & 0.164 & 0.255 & 0.330 & 0.431 & 0.236 & 0.305 & 0.494 & 0.698 \\
      DML2       & 0.113 & 0.153 & 0.290 & 0.339 & 0.123 & 0.201 & 0.631 & 0.369 \\
      AML        & 0.115 & 0.169 & 0.170 & 0.361 & 0.098 & 0.156 & 0.112 & 0.168 \\
      ARB        & 0.094 & 0.141 & 0.210 & 0.315 & 0.095 & 0.158 & 0.113 & 0.191 \\
      DB1        & 0.126 & 0.162 & 0.202 & 0.291 & 0.125 & 0.186 & 0.146 & 0.204 \\
      DB2        & 0.087 & 0.123 & 0.144 & 0.221 & 0.098 & 0.146 & 0.110 & 0.164 \\
      DB6        & 0.067 & 0.089 & 0.116 & 0.161 & 0.070 & 0.099 & 0.078 & 0.108 \\
      \hline\hline
    \end{tabular*}
  \end{center}
  \footnotesize
  \renewcommand{\baselineskip}{11pt}
  \textbf{Note:} 
 $n=500$, $p=800$, $\rho=0.5$, $n_{\mathrm{sim}}=1000$. DML1 = 2-fold sample-split; DML2 = 2-fold cross-fit. AML = 2-fold cross-fit. DB1 = proposed (2-fold sample-split), DB2 = proposed (2-fold cross-fit), DB6 = proposed (no split).
\end{table}
According to Table 1, DB6 attains the lowest MSE in all cells; DB2 is second. When the propensity model is complex (dense $\beta_D$), IPW and DML underperform the simple Regression imputation, consistent with instability from inverse weights. Estimators that avoid inverse weights such as ARB, AML, and ours, are more stable. However,  when the linear approximation is poor, as shown in the supplementary materials, ARB approach deteriorates due to misspecification. DB2 and DB6 outperform AML, reflecting our method’s explicit variance control in addition to bias reduction.

In addition, according to the coverage table in the supplementary materials, confidence intervals based on Theorem~\ref{thm:formalCF} achieve coverage close to 95\% and improve with $n$; DB6 has the shortest intervals.

\section{Empirical Illustration}
\setcounter{equation}{0}
\setcounter{theorem}{0}
\setcounter{assumption}{0}
\setcounter{proposition}{0}
\setcounter{corollary}{0}
\setcounter{lemma}{0}
\setcounter{example}{0}
\setcounter{remark}{0}
We reanalyse the NSW dataset of \citet{dehejia1999causal}. The treated group comprises 185 men randomized to job training between December 1975 and January 1978; the control group is a non-experimental PSID comparison sample of 2,490 men. We define the binary outcome $Y_i$ to indicate whether an individual’s 1978 earnings exceed zero (i.e.\ employment status). To adjust for confounding, we construct a covariate vector $X_i$ of 60 features, including age, years of education, indicators for Black and Hispanic, an indicator for marital status, and earnings in 1974 and 1975, together with their polynomial and interaction terms. We assume unconfoundedness conditional on $X_i$.
\begin{table}
    \caption{Point estimates of ATET in NSW application}
    \begin{center}
    \begin{tabular*}{\textwidth}{@{}lccccccc@{}}
    \hline 
    \hline
       Method & Na\"{i}ve  & Regression & IPW & DML2 & AML & ARB & DB6 \\
       \hline
       Estimates & -0.1585  & 0.2523 & 0.1174 & -0.5635 & 0.0587 & 0.1507  & 0.1122 \\
       \hline \hline
    \end{tabular*}
    \end{center}
    \label{tab:EmpIll}
    \footnotesize
  \renewcommand{\baselineskip}{11pt}
  \textbf{Note:} See note in Table 1.
\end{table}
Table \ref{tab:EmpIll} reports the estimates of ATET using different methods. 
For reference, the simple difference‐in‐means estimate within the randomized experiment is 
0.1106. Our estimator without splitting (DB6) is 0.1122, closest to this benchmark; IPW is 0.1174 and also close. The other approaches deviate more substantially from the experimental result. IPW performs well here because the logistic propensity appears well specified, as discussed by \citet{dehejia1999causal}, and trimming extreme weights has little effect on the estimate.

Moreover, the 95 \% confidence interval for our DB6 estimator is 
[0.0301,0.1944], nearly identical to the randomized‐experiment interval 
[0.0253,0.1959]. This suggests that our method delivers both accurate point estimates and valid inference when applied to non‐experimental data.

\section{Conclusion}
This paper proposes a debiased estimator for causal effects in high-dimensional GLMs with binary outcomes. The estimator augments a regularized GLM plug-in with weights from a single convex program that balances link-derivative–weighted covariates and controls variance, avoiding the inversion of estimated propensity scores. We establish $\sqrt n$-consistency, asymptotic normality, and a feasible variance estimator.
Our analysis focuses on ATET under unconfoundedness for $Y(0)$ and a GLM for $Y(0)\mid X$. Analogous estimators for ATE and ATEC are immediate, but identification then requires $D\!\perp\! Y(1)\mid X$ and a GLM for $Y(1)\mid X$. Sample splitting is used for theory; simulations indicate that the no-split implementation performs best, with cross-fitting close behind.

Compared to doubly robust strategies such as \citet{farrell2015robust}, \citet{belloni2017program}, and \citet{chernozhukov2018double} that utilize the estimation of both outcome and propensity score models, our method leans on the structural assumption of generalized linearity in the outcome model. 
While their methods offer robustness to misspecifications, our approach could serve as a viable alternative when generalized linearity is believed, which is not a overly restrictive assumption under the high-dimensional context, as it circumvents the instability of inverse propensity score estimates.

Two directions could be useful for future work: a data-driven choice of the bias–variance tuning parameter and extensions that relax the known-link assumption (e.g., single-index links) while retaining the balancing-and-variance optimization.

\bibliographystyle{chicago}
\bibliography{ref.bib}

\section*{Appendix A: Proofs of Results}
\renewcommand{\theequation}{A.\arabic{equation}}
\renewcommand{\thesection}{A}
\setcounter{equation}{0}
\setcounter{theorem}{0}
\setcounter{assumption}{0}
\setcounter{proposition}{0}
\setcounter{corollary}{0}
\setcounter{lemma}{0}
\setcounter{example}{0}
\setcounter{remark}{0}
\medskip
\textbf{Proof of Lemma 3.1: }
For each $j \in \{1,...,p\}$, let
\[\left[\sum_{i \in \mathcal{S}_2}\gamma^*_ig'(X_i^\top\hat \beta_{c1})X_i\right]_j = \frac{1}{n_c}\sum_{i\in \mathcal{S}_2} Q_i^\top A_j^i Q_i\]
where we define $A_j^i := g(X_i^\top \hat{\beta}_{c1}) \Sigma^{1/2}e_j \xi^\top \Gamma^{-1} \Sigma^{1/2}$ and $e_j$ denotes the $j$-th basis vector. 
Observe that $A_j^i$ is a rank-1 matrix and its Frobenius norm is given by
\[||A_j^i||_F^2 = \text{tr}\left[(g'(X_i^\top \hat{\beta}_{c1}))^2 \Sigma^{1/2}e_j\xi^\top \Gamma^{-1}\Sigma \Gamma^{-1}\xi e_j^\top \Sigma^{1/2}\right] \leq C,\]
where $C$ is a constant. This inequality holds for all $i \in \mathcal{S}_2$, since $g'(x)$ is bounded for all $x$ according to \textit{Assumption 2.3}, and due to \textit{Assumption 3.1} coupled with the boundness of $||\xi||_2$. 

Under \textit{Assumption 3.1}, and utilizing the Hanson-Wright inequality as described by \citet{rudelson2013hanson}, we can establish that the expression 
$Q_j^\top A_j^i Q_i$ possesses a sub-Exponential distribution. There exists some constants $C_5$ and $C_6$ such that
\[\E \left[exp\{t(Q_j^\top A_j^i Q_i - \E (Q_j^\top A_j^i Q_i))\}\right] \leq exp\{C_5t^2c^4C\} \text{ for all } t \leq \frac{C_6}{c^2\sqrt{C}}.\]
Given $\E[\sum_{i \in \mathcal{S}_2}\gamma^*_ig'(X_i^\top\hat \beta_{c1})X_i] = \xi$, for any sequence $t_n >0$ with $t_n^2/n_c \to 0$,
\[\E \left[exp\left\{\sqrt{n_c}t_n\left(\sum_{i \in \mathcal{S}_2}\tilde{\gamma}^*_iw_ig'(X_i^\top\hat \beta_{c1})X_i-\xi\right)_j\right\}\right] \leq exp\{C_5t^2c^4VC\} \]
holds if $n_c$ is sufficiently large. By setting $t_n:= \sqrt{\log (p/2\delta)}/(c^2\sqrt{C_5VS})$ and utilizing a symmetric argument, we arrive at
\[\mathbb{P} \left[\left|\sqrt{n_c}\left(\sum_{i \in \mathcal{S}_2} \gamma^*_ig'(X_i^\top\hat \beta_{c1})X_i-\xi\right)_j\right| > 2c^2\sqrt{C_5VC\log\left(\frac{p}{2\delta}\right)}\right] \leq \delta/p\]
Applying a union bound, we deduce
\[ \left\Vert \xi -  \sum_{\{i \in \mathcal{S}_2\}} \gamma_i^* g'(X_i^\top\hat \beta_{c1})X_i \right \Vert_\infty \lesssim  \sqrt{\log p/n_c}\]
which holds with probability approaching one. Additionally, by \textit{Assumption 3.1},
each $n_c\gamma_i^*$ follows a sub-Gaussian distribution with constant variance proxy, leading to $||n_c\boldsymbol{\gamma}^*||_\infty=O_p (\sqrt{\log n_c})$ and consequently $||\boldsymbol{\gamma}^*||_\infty=O_p (\frac{\sqrt{\log n_c}}{n_c})$.

\medskip

\textbf{Proof of Theorem 3.1: }
In the course of proving Theorem 3.1, we will rely on the subsequent Lemma \ref{lemma:riskbound}. 
\renewcommand{\thelemma}{A.\arabic{lemma}}
\renewcommand{\thesection}{A}
\setcounter{lemma}{0}
\begin{lemma}
    Under \textit{Assumption 2.1,2.2,3.2,3.3} and assuming that $k = O_p(n_c/\log p)$, the following inequalities hold:
    \begin{align*}
        ||\hat{\beta}_{c1}-\beta_{c}||_1 = O_p\left(k\sqrt{\frac{\log p}{n_c}}\right),\
        ||\hat{\beta}_{c1}-\beta_{c}||_2 = O_p\left(\sqrt{\frac{k\log p}{n_c}}\right).
    \end{align*}
    \label{lemma:riskbound}
\end{lemma}

\begin{proof}
    The proof for Lemma \ref{lemma:riskbound} is omitted, as it can be derived from the logic presented in the proof of Theorem 2 by \citet{cai2023statistical}. 
However, it is vital to note the substitution of their assumptions (L2) and (L3) with our particular \textit{Assumption 2.3(b)} and \textit{2.3(c)}. 
We shall now advance to the proof of Theorem 3.1.
\end{proof}

    The estimation error for the estimator $\hat{\theta}$ can be expressed as
    \begin{equation}
        \hat{\theta} - \theta = \left(\xi -  \sum_{\{i \in \mathcal{S}_2\}} \gamma_i g'(X_i^\top\hat \beta_{c1})X_i\right)^\top(\hat{\beta}_{c1} - \beta_c)+\sum_{i\in \mathcal{S}_2} \gamma_i \Delta_i^c + \sum_{i\in \mathcal{S}_2} \gamma_i\varepsilon_i
    \end{equation}
    where $\Delta_i^c = g''(X_i^\top\hat \beta_{c1} + t_{i}^{c}X_i^\top(\beta-\hat \beta_{c1}))[X_i^\top(\hat \beta_{c1}-\beta)]^2$ for some $t_{i}^{c} \in (0,1)$.\\
    \\
    Because \textit{Assumption 3.3} is more relaxed than \textit{Assumption 3.1}, the estimation error bound of Lemma \ref{lemma:riskbound} also holds under \textit{Assumption 2.2-3.2}.\\
    \\
    (1) \textbf{Concerning the First Term}: We aim to demonstrate that the first term diminishes faster than $1/\sqrt{n_c}$.
    \medskip

    Lemma 3.1 provides
    $\left\Vert \xi -  \sum_{\{i \in \mathcal{S}_2\}} \gamma_i g'(X_i^\top\hat \beta_{c1})X_i\right\Vert_\infty = O_p \left(\sqrt{\frac{\log p}{n_c}}\right)$.
    With the assumptions and Lemma \ref{lemma:riskbound}, we have the $L_1$-risk bound for the lasso estimator
    $||\hat \beta_{c1}-\beta||_1 = O_p \left(k\sqrt{\frac{\log p}{n_c}}\right)$. Hence, we obtain
    \[\left|\left(\xi -  \sum_{\{i \in \mathcal{S}_2\}} \gamma_i g'(X_i^\top\hat \beta_{c1})X_i\right)^\top (\hat{\beta}_{c1} - \beta_c)\right| \leq O_p \left(\frac{k\log p}{n_c}\right),\]
    which vanishes at a rate faster than $1/\sqrt{n_c}$ given $k \ll \frac{\sqrt{n_c}}{\log p}$. \\
    \\
    (2) \textbf{Concerning the Second Term}: We aim to demonstrate that the second term decays faster than $1/\sqrt{n_c}$.
    \medskip
    
    By H\"older inequality, we have
    \[\left\vert\sum_{i\in \mathcal{S}_2} \gamma_i \Delta_i^c \right\vert \leq \max_{i\in \mathcal{S}_2} |\gamma_i| \cdot \sum_{i\in \mathcal{S}_2} |\Delta_i^c|\]
    Lemma 3.1 shows $ \max_{i\in \mathcal{S}_2} |\gamma_i| = O_p (\frac{\sqrt{\log n_c}}{n_c})$. 
    Using \textit{Assumption 2.3(c)}, it follows that
    \begin{equation*}
        \begin{aligned}
            \sum_{i\in \mathcal{S}_2} |\Delta_i^c| & \leq \sum_{i\in \mathcal{S}_2} \left|g''(X_i^\top\hat \beta_{c1} + t_{i}^{c}X_i^\top(\beta-\hat \beta_{c1}))\right|[X_i^\top(\hat \beta_{c1}-\beta)]^2 \\
            & \lesssim \sum_{i\in \mathcal{S}_2}[X_i^\top(\hat \beta_{c1}-\beta)]^2 =O_p \left(k\log p \right).
        \end{aligned}
    \end{equation*}
    because $\frac{1}{n_c}\sum_{i\in \mathcal{S}_2}[X_i^\top(\hat \beta_{c1}-\beta)]^2 = (\hat{\beta}_{c1} - \beta)^\top  (\frac{1}{n_c}\sum_{i\in \mathcal{S}_2}X_i X_i^\top) \cdot (\hat{\beta}_{c1} - \beta) \lesssim O_p (k\log p/n_c) $ by \textit{Assumption 3.1} and Lemma \ref{lemma:riskbound}. 
    Therefore, we can conclude that $\left\vert\sum_{i\in \mathcal{S}_2} \gamma_i \Delta_i^c \right\vert$ diminishes at a rate faster than $1/\sqrt{n_c}$ since $k \ll \frac{\sqrt{n_c}}{\log p\sqrt{\log n_c}}$. \\
    \\
    (3) \textbf{Concerning the Third Term}: We aim to demonstrate that the third term converges to a Gaussian distribution at a rate of $\sqrt{n_c}$. For this purpose, the following lemma, which is proved in Online Supplement, is required:

    \begin{lemma}
        Under the conditions of Theorem 1, it holds that $ v = O_p (1/n_c)$, $\hat{v} = O_p (1/n_c)$, and $|\hat{v} - v| = o_p (1/n_c)$.
    \label{lemma:Bound4v}
    \end{lemma}

    By Lemma 3.1 and the feasibility of $\gamma_i$, we have 
    \begin{equation*}
        \max_{i \in \mathcal{S}_2} |\gamma_i| \lesssim O_p\left( \frac{\sqrt{\log n_c}}{n_c} \right),
    \end{equation*}
    Because $\gamma_i$ is independent of $\varepsilon_i$ conditional on $X_i$,
    \[\E\left[\sum_{i \in \mathcal{S}_2} \gamma_i \varepsilon_i | \{X_i\}_{i \in \mathcal{S}_2}\right] = 0,\]
    \[\text{var}\left[\sum_{i \in \mathcal{S}_2} \gamma_i \varepsilon_i | \{X_i\}_{i \in \mathcal{S}_2}\right] = \sum_{i \in \mathcal{S}_2} \gamma_i ^2 g(X_i^\top \beta_c)(1-g(X_i^\top \beta_c)) = v,\]
    \begin{equation*}
            \E \left[\sum_{i \in \mathcal{S}_2} \left(\gamma_i \varepsilon_i\right)^3 | \{X_i\}_{i \in \mathcal{S}_2}\right] \leq C_7 \sum_{i \in \mathcal{S}_2} \gamma_i^3 g(X_i^\top \beta_c)(1-g(X_i^\top \beta_c)) \\
    \lesssim \left(\max_{i\in \mathcal{S}_2}|\gamma_i|\right) v,
    \end{equation*}
    for some constant $C_7$, where the first inequality follows by sub-Gaussianity of $\varepsilon_i$ and $g(x) \in (0,1)$ for all $x$, and the last inequality follows because of the feasibility of $\boldsymbol{\gamma}$.
    So we have
    \[\frac{\E \left[\sum_{i \in \mathcal{S}_2} \left(\gamma_i g'(X_i^\top \hat{\beta}_{c1}) \varepsilon_i\right)^3 | \{X_i\}_{i \in \mathcal{S}_2}\right] }{\left\{\text{var}\left[\sum_{i \in \mathcal{S}_2} \gamma_i g'(X_i^\top \hat{\beta}_{c1}) \varepsilon_i | \{X_i\}_{i \in \mathcal{S}_2}\right]\right\}^{\frac{3}{2}}} \lesssim \frac{\left(\max_{i\in \mathcal{S}_2}|\gamma_i|\right)}{\sqrt{v}} \lesssim \sqrt{\frac{\log n_c}{n_c}}\ \to 0,\]
    where the last inequality follows by Lemma \ref{lemma:Bound4v}. Thus the Lyapunov's condition gives that
    \[\frac{\hat{\theta} - \theta}{\sqrt{v}} \Rightarrow N(0,1)\]
    In addition, Lemma \ref{lemma:Bound4v} tells that $\frac{\hat{\theta} - \theta}{\sqrt{\hat{v}}} \Rightarrow N(0,1)$ by Slutsky theorem.

\medskip
\textbf{Proof of Theorem 3.2: }
We need the following Lemma \ref{lemma:pscore}, which is proved in Online Supplement, to complete this proof.
    \begin{lemma}
        Under the conditions of Theorem 3.2, the weights $\gamma^{**}$ defined by
        \[\boldsymbol{\gamma}^{**} = \frac{\frac{p(X_i)}{1-p(X_i)}}{\sum_{i \in \mathcal{S}_2}\frac{p(X_i)}{1-p(X_i)}},\]
        satisfy the following bounds with probability tending to 1:
        \[\left\Vert \frac{1}{n_t}\sum_{\{i:D_i=1\}} g'(X_i^\top\hat \beta_{c1})X_i^\top -  \sum_{\{i \in \mathcal{S}_2\}} \gamma_i^{**} g'(X_i^\top\hat \beta_{c1})X_i^\top \right \Vert_\infty \lesssim \sqrt{\log p/n_c},\]
        \[n_c||\boldsymbol{\gamma}^{**}||_2^2 \leq \frac{1}{\rho^2}\E\left[\frac{p(X_i)^2}{(1-p(X_i))^2}|D_i = 0\right] + o_p(1).\]
    
        \label{lemma:pscore}
    \end{lemma}

    We can find a feasible weight to the optimization problem:
    \[\gamma_i^{**} = \frac{\frac{p(X_i)}{1-p(X_i)}}{\sum_{i \in \mathcal{S}_2}\frac{p(X_i)}{1-p(X_i)}}.\]
    By the definition and overlap assumption, it is trivial to show that $\sum_{i=1}^n \gamma_i^{**} = 1$ and $0 < \gamma^{**}_i \lesssim O_p(\sqrt{\log n_c}/n_c)$. The feasibility follows by Lemma \ref{lemma:pscore}. Then to show Theorem 3.2, we can use similar arguments in Theorem 3.1. 

        The estimation error for the estimator $\hat{\mu}$ can be decomposed as 
    \begin{equation}
        \begin{aligned}
            \hat{\mu}_c - \mu_c & = \left(\frac{1}{n_t}\sum_{\{i:D_i=1\}} g'(X_i^\top\hat \beta_{c1})X_i -  \sum_{\{i \in \mathcal{S}_2\}} \gamma_i g'(X_i^\top\hat \beta_{c1})X_i \right)^\top (\hat{\beta}_{c1} - \beta_c)\\
            &+\sum_{i\in \mathcal{S}_2} \gamma_i \Delta_i^c 
            + \frac{1}{n_t} \sum_{\{i:D_i=1\}} \Delta_i^t  + \sum_{i\in \mathcal{S}_2} \gamma_i\varepsilon_i
        \end{aligned}
    \end{equation}
    where $\Delta_i^c = g''(X_i^\top\hat \beta_{c1} + t_{i}^{c}X_i^\top(\beta-\hat \beta_{c1}))[X_i^\top(\hat \beta_{c1}-\beta)]^2$ for some $t_{i}^{c} \in (0,1)$ and 
    $\Delta_i^t = g''(X_i^\top\hat \beta_{c1} + t_{i}^{t}X_i^\top(\beta-\hat \beta_{c1}))[X_i^\top(\hat \beta_{c1}-\beta)]^2$ for some $t_{i}^{t} \in (0,1)$.\\
    \\
    (1) \textbf{Concerning the First Term}: We aim to demonstrate that the first term diminishes faster than $1/\sqrt{n_c}$.
    \medskip

    Lemma \ref{lemma:pscore} implies the feasibility of the optimization problem and thus
    $$\left\Vert \frac{1}{n_t}\sum_{\{i:D_i=1\}} g'(X_i^\top\hat \beta_{c1})X_i -  \sum_{\{i \in \mathcal{S}_2\}} \gamma_i g'(X_i^\top\hat \beta_{c1})X_i\right\Vert_\infty = O_p \left(\sqrt{\frac{\log p}{n_c}}\right).$$
    With the assumptions and Lemma \ref{lemma:riskbound}, we have the $L_1$-risk bound for the lasso estimator
    $||\hat \beta_{c1}-\beta||_1 = O_p \left(k\sqrt{\frac{\log p}{n_c}}\right)$. Hence, we obtain
    \[\left|\left(\frac{1}{n_t}\sum_{\{i:D_i=1\}} g'(X_i^\top\hat \beta_{c1})X_i -  \sum_{\{i \in \mathcal{S}_2\}} \gamma_i g'(X_i^\top\hat \beta_{c1})X_i\right)^\top (\hat{\beta}_{c1} - \beta_c)\right| \leq O_p \left(\frac{k\log p}{n_c}\right),\]
    which vanishes at a rate faster than $1/\sqrt{n_c}$ given $k \ll \frac{\sqrt{n_c}}{\log p}$. \\
    \\
    (2) \textbf{Concerning the Second Term}: We aim to show that the second term decays faster than $1/\sqrt{n_c}$.
    \medskip
    
    By H\"older inequality, we have
    \[\left\vert\sum_{i\in \mathcal{S}_2} \gamma_i \Delta_i^c \right\vert \leq \max_{i\in \mathcal{S}_2} |\gamma_i| \cdot \sum_{i\in \mathcal{S}_2} |\Delta_i^c|\]
    The feasibility of the optimization problem implies $ \max_{i\in \mathcal{S}_2} |\gamma_i| = O_p (\frac{\sqrt{\log n_c}}{n_c})$. 
    Using \textit{Assumption 2.3(c)}, it follows that
    \begin{equation*}
        \begin{aligned}
            \sum_{i\in \mathcal{S}_2} |\Delta_i^c| & \leq \sum_{i\in \mathcal{S}_2} \left|g''(X_i^\top\hat \beta_{c1} + t_{i}^{c}X_i^\top(\beta-\hat \beta_{c1}))\right|[X_i^\top(\hat \beta_{c1}-\beta)]^2 \\
            & \lesssim \sum_{i\in \mathcal{S}_2}[X_i^\top(\hat \beta_{c1}-\beta)]^2 =O_p \left(k\log p \right).
        \end{aligned}
    \end{equation*}
    because $\frac{1}{n_c}\sum_{i\in \mathcal{S}_2}[X_i^\top(\hat \beta_{c1}-\beta)]^2 = (\hat{\beta}_{c1} - \beta)^\top  (\frac{1}{n_c}\sum_{i\in \mathcal{S}_2}X_i X_i^\top) \cdot (\hat{\beta}_{c1} - \beta) \lesssim O_p (k\log p/n_c) $ by \textit{Assumption 3.3} and Lemma \ref{lemma:riskbound}. 
    Therefore, we can conclude that $\left\vert\sum_{i\in \mathcal{S}_2} \gamma_i \Delta_i^c \right\vert$ diminishes at a rate faster than $1/\sqrt{n_c}$ since $k \ll \frac{\sqrt{n_c}}{\log p\sqrt{\log n_c}}$. \\
    \\
    (3) \textbf{Concerning the Third Term}: We aim to demonstrate that the third term decays faster than $1/\sqrt{n_c}$.
    \medskip
    
    By Triangle inequality, we have
    \[\left\vert\sum_{\{i:D_i=1\}} \Delta_i^t \right\vert \leq \sum_{\{i:D_i=1\}} |\Delta_i^t|\]
    Using \textit{Assumption A 3.3}, it follows that
    \begin{equation*}
        \begin{aligned}
            \sum_{\{i:D_i=1\}} |\Delta_i^t| & \leq \sum_{\{i:D_i=1\}} \left|g''(X_i^\top\hat \beta_{c1} + t_{i}^{t}X_i^\top(\beta-\hat \beta_{c1}))\right|[X_i^\top(\hat \beta_{c1}-\beta)]^2 \\
            & \lesssim \sum_{\{i:D_i=1\}}[X_i^\top(\hat \beta_{c1}-\beta)]^2 =O_p \left(k\log p \right).
        \end{aligned}
    \end{equation*}
    because $\frac{1}{n_t}\sum_{\{i:D_i=1\}}[X_i^\top(\hat \beta_{c1}-\beta)]^2 = (\hat{\beta}_{c1} - \beta)^\top  (\frac{1}{n_t}\sum_{\{i:D_i=1\}}X_i X_i^\top) \cdot (\hat{\beta}_{c1} - \beta) \lesssim O_p (k\log p/n_c) $ by \textit{Assumptions 3.3-3.4} and Lemma \ref{lemma:riskbound}. 
    Therefore, we can conclude that $\left\vert\sum_{\{i:D_i=1\}} \Delta_i^t \right\vert$ diminishes at a rate faster than $1/\sqrt{n_c}$ since $k \ll \frac{\sqrt{n_c}}{\log p}$. \\
    \\
    (4) \textbf{Concerning the Fourth Term}: We aim to demonstrate that the fourth term converges to a Gaussian distribution at the rate of $\sqrt{n_c}$.
    \medskip

    The feasibility of $\gamma_i$, we have 
    \begin{equation*}
        \max_{i \in \mathcal{S}_2} |\gamma_i| \lesssim O_p\left( \frac{\sqrt{\log n_c}}{n_c} \right),
    \end{equation*}
    Because $\gamma_i$ is independent of $\varepsilon_i$ conditional on $X_i$,
    \[\E\left[\sum_{i \in \mathcal{S}_2} \gamma_i \varepsilon_i | \{X_i\}_{i \in \mathcal{S}_2}\right] = 0,\]
    \[\text{var}\left[\sum_{i \in \mathcal{S}_2} \gamma_i \varepsilon_i | \{X_i\}_{i \in \mathcal{S}_2}\right] = \sum_{i \in \mathcal{S}_2} \gamma_i ^2 g(X_i^\top \beta_c)(1-g(X_i^\top \beta_c)) = v,\]
    \begin{equation*}
            \E \left[\sum_{i \in \mathcal{S}_2} \left(\gamma_i \varepsilon_i\right)^3 | \{X_i\}_{i \in \mathcal{S}_2}\right] \leq C_7 \sum_{i \in \mathcal{S}_2} \gamma_i^3 g(X_i^\top \beta_c)(1-g(X_i^\top \beta_c)) \\
    \lesssim \left(\max_{i\in \mathcal{S}_2}|\gamma_i|\right) v,
    \end{equation*}
    for some constant $C_7$, where the first inequality follows by sub-Gaussianity of $\varepsilon_i$ and $g(x) \in (0,1)$ for all $x$, and the last inequality follows because of the feasibility of $\boldsymbol{\gamma}$.
    So we have
    \[\frac{\E \left[\sum_{i \in \mathcal{S}_2} \left(\gamma_i g'(X_i^\top \hat{\beta}_{c1}) \varepsilon_i\right)^3 | \{X_i\}_{i \in \mathcal{S}_2}\right] }{\left\{\text{var}\left[\sum_{i \in \mathcal{S}_2} \gamma_i g'(X_i^\top \hat{\beta}_{c1}) \varepsilon_i | \{X_i\}_{i \in \mathcal{S}_2}\right]\right\}^{\frac{3}{2}}} \lesssim \frac{\left(\max_{i\in \mathcal{S}_2}|\gamma_i|\right)}{\sqrt{v}} \lesssim \sqrt{\frac{\log n_c}{n_c}}\ \to 0,\]
    where the last inequality follows by Lemma \ref{lemma:Bound4v_2}, which is proved in Online Supplement:
    \begin{lemma}
        Under the conditions of Theorem 3.2, it holds that $v = O_p(1/n_c)$, $\hat{v} = O_p(1/n_c)$ and $|\hat{v}-v| = o_p(1/n_c)$.
        \label{lemma:Bound4v_2}
    \end{lemma}
    Thus the Lyapunov's condition gives that
    \[\frac{\hat{\theta} - \theta}{\sqrt{v}} \Rightarrow N(0,1)\]
    In addition, Lemma \ref{lemma:Bound4v_2} tells that $\frac{\hat{\theta}-\theta}{\sqrt{\hat{v}}} \Rightarrow N(0,1)$ by Slutsky theorem.

\clearpage

% put S in front of counters for online supplement
\renewcommand{\theequation}{S.\arabic{equation}}
\renewcommand{\thetable}{S.\arabic{table}}
\renewcommand{\thesection}{S\arabic{section}}
\renewcommand{\thepage}{S\arabic{page}}

\setcounter{equation}{0}
\setcounter{section}{0}
\setcounter{page}{1}

\vspace{0.2cm}
\begin{center}
\textbf{\large
    Online Supplement to ``Causal Inference in High‐dimensional Generalized Linear Models with Binary Outcomes"}
\end{center}
\vspace{0.2cm}

\section{ADDITIONAL PROOFS}
\renewcommand{\theequation}{S.\arabic{equation}}
\setcounter{equation}{0}
\medskip
\textbf{Proof of Lemma A.3: }
We first want to show that
    \[\max_{i \in \mathcal{S}_2} \left|\frac{g(|X_i^\top\beta_c|)}{g(|X_i^\top\hat{\beta}_{c1}|)} - 1\right| = o_p(1), \quad \max_{i \in \mathcal{S}_2} \left|\frac{1-g(|X_i^\top\beta_c|)}{1-g(|X_i^\top\hat{\beta}_{c1}|)} - 1\right| = o_p(1)\]
    By Taylor expansion, we have 
    \[g(|X_i^\top\beta_c|) = g(|X_i^\top\hat{\beta}_{c1}|) + g'(a)\left|X_i^\top (\hat{\beta}_{c1} - \beta_c)\right|,\]
    where $a$ is between $|X_i^\top\beta_c|$ and $|X_i^\top\hat{\beta}_{c1}|,$
    So, by \textit{Assumption 2.3(b)}, we can get
    \[\max_{i \in \mathcal{S}_2} \left|\frac{g(|X_i^\top\beta_c|)}{g(|X_i^\top\hat{\beta}_{c1}|)} - 1\right| \lesssim \max_{i \in \mathcal{S}_2}\left|X_i^\top (\hat{\beta}_{c1} - \beta_c)\right| \lesssim O_p \left(\sqrt{\frac{k\log p \log n_c}{n_c}}\right),\]
    and thus it goes to zero by $k \ll \frac{n_c}{\log p \log n_c}$. In addition, the other equality follows by the same argument.

    Because there exists an index $j \in \{1,...,p\}$ such that $|\xi_j| \geq \kappa$, any feasible solution to the optimization problem (19) must eventually satisfy $\left[\sum_{i \in \mathcal{S}_2} \gamma_i g'(X_i^\top\hat{\beta}_{c1})X_i\right]^2_j \geq \kappa^2/2$.
    By Cauchy-Schwarz inequality, we have
            \[\sum_{i \in \mathcal{S}_2} \left[\gamma_i g'(X_i^\top\hat{\beta}_{c1})\right]^2 \geq \kappa^2/\left(2\sum_{i \in \mathcal{S}_2} X_{ij}^2\right) = O_p \left(\frac{1}{n_c}\right)\] 
            Thus we have
            \[O_p \left(\frac{1}{n_c}\right) \leq \sum_{i \in \mathcal{S}_2} \left[\gamma_i g'(X_i^\top\hat{\beta}_{c1})\right]^2 = \sum_{i \in \mathcal{S}_2} \gamma_i^2 g(X_i^\top\hat{\beta}_{c1})(1-g(X_i^\top\hat{\beta}_{c1}))\frac{[g'(X_i^\top\hat{\beta}_{c1})]^2}{g(X_i^\top\hat{\beta}_{c1})(1-g(X_i^\top\hat{\beta}_{c1}))}, \]
    Let that \[\hat{v} := \sum_{i \in \mathcal{S}_2} \gamma_i^2 g(X_i^\top\hat{\beta}_{c1})(1-g(X_i^\top\hat{\beta}_{c1})),\] and it is upper bounded by
    \[\hat{v} \leq \hat{v}^*:= \sum_{i \in \mathcal{S}_2} (\tilde{\gamma_i}^*)^2 w_i^2 g(X_i^\top\hat{\beta}_{c1})(1-g(X_i^\top\hat{\beta}_{c1})) = \sum_{i \in \mathcal{S}_2} (\tilde{\gamma_i}^*)^2 h(X_i^\top\hat{\beta}_{c1}) \lesssim O_p\left(\frac{1}{n_c}\right), \]
            because $\tilde{\gamma}_i^* = \frac{1}{n_c}\xi^\top \Gamma^{-1} X_i$ and the sub-Gaussianity of $X_i$. Because 
            \[\frac{[g'(X_i^\top\hat{\beta}_{c1})]^2}{g(X_i^\top\hat{\beta}_{c1})(1-g(X_i^\top\hat{\beta}_{c1}))} \leq \frac{[g'(|X_i^\top\hat{\beta}_{c1}|)]^2}{1-g(|X_i^\top\hat{\beta}_{c1}|)} \leq C,\]
            for some constant C by \textit{Assumption 2.3}, we have
            \[\hat{v} = O_p\left(\frac{1}{n_c}\right).\]
            Next, we want to find the asymptotic rate of the difference between $v$ and $\hat{v}$,
            \begin{equation}
                \begin{aligned}
                    |\hat{v} - v| &= \left|\sum_{i \in \mathcal{S}_2} \gamma_i^2\left[g(X_i^\top\hat{\beta}_{c1})(1-g(X_i^\top\hat{\beta}_{c1})) - g(X_i^\top \beta_c)(1-g(X_i^\top \beta_c))\right]\right| \\
                    & \leq \hat{v} \cdot \max_{i \in \mathcal{S}_2} \left|1-\frac{g(X_i^\top\beta_c)(1-g(X_i^\top\beta_c))}{g(X_i^\top\hat{\beta}_{c1})(1-g(X_i^\top\hat{\beta}_{c1}))}\right| \\
                    & \lesssim o_p\left(\frac{1}{n_c}\right).
                   \label{eq:vhatdiffv}
                \end{aligned}
            \end{equation}
            Thus we have
            \[v = O_p\left(\frac{1}{n_c}\right).\]

\medskip

\textbf{Proof of Lemma A.4: }
Our main quantity of interest $\frac{1}{n_t}\sum_{\{i:D_i=1\}} g'(X_i^\top\hat \beta_{c1})X_i^\top -  \sum_{\{i \in \mathcal{S}_2\}} \gamma_i g'(X_i^\top\hat \beta_{c1})X_i^\top$ exhibits translation invariance. Specifically, the transformation $g'(X_i^\top\hat \beta_{c1})X_i \to g'(X_i^\top\hat \beta_{c1})X_i+c$ for any constant $c \in \mathbb{R}$ retains the quantity's value. 
So without loss of generality, re-centre the problem such that the mean $\E(g'(X_i^\top\hat{\beta}_{c1})X_{ij}|D = 1) = 0$. Then for each $j = 1,...,p$, there exists constants $c_1, c_2, c_3, c_4>0$ such that 
\setcounter{beana}{0}
\begin{list}
{(\alph{beana})}{\usecounter{beana}}
        \item $\frac{1}{n_t}\sum_{\{i:D_i=1\}} g'(X_i^\top\hat \beta_{c1})X_{ij}$ is sub-Gaussian with parameter $c_1/n_t$ by \textit{Assumption 3.3};
        \item $A_j:= \frac{1}{n_c}\sum_{\{i \in \mathcal{S}_2\}} g'(X_i^\top\hat \beta_{c1})X_{ij}\frac{p(X_i)}{1-p(X_i)}$ is sub-Gaussian with parameter $c_2/n_c$ by the boundedness of $g'(x)$ for all $x$ and $\zeta \leq p(X_i) \leq 1-\zeta $ for all $i$. Note that $\E(A_j) = \E(g'(X_i^\top\hat{\beta}_{c1})X_{ij}|D = 1) = 0$;
        \item $b_1:=\frac{1}{n_c}\sum_{\{i \in \mathcal{S}_2\}} \frac{p(X_i)}{1-p(X_i)}$ is sub-Gaussian with parameter $c_3/n_c$ after re-centring;
        \item $b_2:=\frac{1}{n_c}\sum_{\{i \in \mathcal{S}_2\}} \left(\frac{p(X_i)}{1-p(X_i)}\right)^2$ is sub-Gaussian with parameter $c_4/n_c$ after re-centring.
    \end{list}
    With a union bound, we can get that 
    \begin{equation*}
        \begin{aligned}
            ||A||_\infty &\lesssim O_p\left(\sqrt{\frac{\log p}{n_c}}\right), \\
            \left\Vert \frac{1}{n_t}\sum_{\{i:D_i=1\}} g'(X_i^\top\hat \beta_{c1})X_i - A \right\Vert_\infty & \lesssim O_p\left(\sqrt{\frac{\log p}{n_c}}\right), \\
            |b_1-\rho| &\lesssim O_p\left(\frac{1}{\sqrt{n_c}}\right), \\
            |b_2-\E[b_2]| & \lesssim O_p\left(\frac{1}{\sqrt{n_c}}\right).
        \end{aligned}
    \end{equation*}
    So we can further get that 
    \begin{equation*}
        \begin{aligned}
            &\left\Vert \frac{1}{n_t}\sum_{\{i:D_i=1\}} g'(X_i^\top\hat \beta_{c1})X_i^\top -  \sum_{\{i \in \mathcal{S}_2\}} \gamma_i g'(X_i^\top\hat \beta_{c1})X_i^\top \right\Vert_\infty  = \left\Vert \frac{1}{n_t}\sum_{\{i:D_i=1\}} g'(X_i^\top\hat \beta_{c1})X_i -  \frac{1}{b_1}A \right\Vert_\infty \\
            &\leq \left\Vert \frac{1}{n_t}\sum_{\{i:D_i=1\}} g'(X_i^\top\hat \beta_{c1})X_i - A \right\Vert_\infty + \left|1-\frac{1}{b_1}\right|||A||_\infty \lesssim O_p\left(\sqrt{\frac{\log p}{n_c}}\right).
        \end{aligned}
    \end{equation*}
    In addition, for the weights $\gamma^{**}$, we have 
    \[n_c||\boldsymbol{\gamma}_i^{**}||_2^2 = \frac{b_2}{b_1^2} \leq \frac{\E[b_2]}{\rho^2} + o_p(1),\]
since $\E[b_2] = \E\left[\frac{p(X_i)^2}{(1-p(X_i))^2}|D_i = 0\right] \in (0, (1-\epsilon^2)/\epsilon^2)$.

\medskip

\textbf{Proof of Lemma A.5: }
Because there exits an index $j \in \{1,...,p\}$ such that $|\E[g'(X^\top \hat{\beta}_{c1})X|D=1]|_j \geq \kappa$, any feasible solution to the optimization problem (6) must eventually satisfy $\left[\sum_{i \in \mathcal{S}_2} \gamma_i g'(X_i^\top\hat{\beta}_{c1})X_i\right]^2_j \geq \kappa^2/2$.
    By Cauchy-Schwarz inequality, we have
            \[\sum_{i \in \mathcal{S}_2} \left[\gamma_i g'(X_i^\top\hat{\beta}_{c1})\right]^2 \geq \kappa^2/\left(2\sum_{i \in \mathcal{S}_2} X_{ij}^2\right) = O_p \left(\frac{1}{n_c}\right)\] 
            Thus we have
            \[O_p \left(\frac{1}{n_c}\right) \leq \sum_{i \in \mathcal{S}_2} \left[\gamma_i g'(X_i^\top\hat{\beta}_{c1})\right]^2 = \sum_{i \in \mathcal{S}_2} \gamma_i^2 g(X_i^\top\hat{\beta}_{c1})(1-g(X_i^\top\hat{\beta}_{c1}))\frac{[g'(X_i^\top\hat{\beta}_{c1})]^2}{g(X_i^\top\hat{\beta}_{c1})(1-g(X_i^\top\hat{\beta}_{c1}))}, \]
            Let that \[\hat{v} := \sum_{i \in \mathcal{S}_2} \gamma_i^2 g(X_i^\top\hat{\beta}_{c1})(1-g(X_i^\top\hat{\beta}_{c1})),\] and it is upper bounded by
            \[\hat{v} \leq \hat{v}^{**}:= \sum_{i\in \mathcal{S}_2} (\gamma_i^{**})^2 g(X_i^\top \hat{\beta}_{c1})(1-g(X_i^\top \hat{\beta}_{c1})) \leq \left[\max_{i \in \mathcal{S}_2} g(X_i^\top \hat{\beta}_{c1})(1-g(X_i^\top \hat{\beta}_{c1}))\right] \cdot ||\boldsymbol{\gamma}^{**}||_2^2 \lesssim O_p(1/n_c), \]
            by Lemma A.4 and the sub-Gaussianity of $X_i$. Because 
            \[\frac{[g'(X_i^\top\hat{\beta}_{c1})]^2}{g(X_i^\top\hat{\beta}_{c1})(1-g(X_i^\top\hat{\beta}_{c1}))} \leq \frac{[g'(|X_i^\top\hat{\beta}_{c1}|)]^2}{1-g(|X_i^\top\hat{\beta}_{c1}|)} \leq C,\]
            for some constant C by \textit{Assumption 2.3}, we have
            \[\hat{v} = O_p\left(\frac{1}{n_c}\right).\]
            Similar to the proof in Section A.3, we can find asymptotic rate of the difference between $v$ and $\hat{v}$,
            \begin{equation}
                \begin{aligned}
                    |\hat{v} - v| &= \left|\sum_{i \in \mathcal{S}_2} \gamma_i^2\left[g(X_i^\top\hat{\beta}_{c1})(1-g(X_i^\top\hat{\beta}_{c1})) - g(X_i^\top \beta_c)(1-g(X_i^\top \beta_c))\right]\right| \\
                    & \leq \hat{v} \cdot \max_{i \in \mathcal{S}_2} \left|1-\frac{g(X_i^\top\beta_c)(1-g(X_i^\top\beta_c))}{g(X_i^\top\hat{\beta}_{c1})(1-g(X_i^\top\hat{\beta}_{c1}))}\right| \\
                    & \lesssim o_p\left(\frac{1}{n_c}\right).
                   \label{eq:vhatdiffv_2}
                \end{aligned}
            \end{equation}
            Thus we have
            \[v = O_p\left(\frac{1}{n_c}\right).\]

\section{Properties of Examples}
In this section, we validate \textit{Assumptions 2.3(a)-(c)} for the link functions in three examples. For the proof of \textit{Assumption 2.3(d)}, refer to the supplementary materials of \citet{cai2023statistical}.

\subsection{Logistic Link Function}
The logistic link function is defined as
\[g(x) = \frac{\exp(x)}{1+\exp(x)},\]
and its derivatives are given by 
\[g'(x) = \frac{\exp(x)}{(1+\exp(x))^2}, \quad g''(x) = \frac{\exp(x)(1-\exp(x))}{(1+\exp(x))^3}.\]
\textit{Assumption 2.3(a)} holds since $0 < g'(x) \leq \frac{1}{4}$ for $x \in \mathbb{R}$ and $g''(x) <0 $ for $x>0$. Next, we want to show that $g(x) \leq \Phi(x)$ for $x \geq 0$.
To achieve this, we define $h_1(x) := g(x) - \Phi(x)$, from which it follows that
\[h_1'(x) = \frac{\sqrt{2\pi}\exp\left(x+\frac{x^2}{2}\right) - (1+\exp(x))^2}{\sqrt{2\pi}(1+\exp(x))^2\exp\left(\frac{x^2}{2}\right)} < 0 \text{ for } x\in [0,1].\]
This inequality is justified by the fact that $x \in [0,1]$
\begin{equation*}
    \begin{aligned}
        \sqrt{2\pi}\exp\left(x+\frac{x^2}{2}\right) - (1+\exp(x))^2 & = \left(\sqrt{2\pi} - \exp(x/2)-2\exp(-x/2)-\exp(-3x/2)\right)\exp(3x/2) \\
        & \leq \left(\sqrt{2\pi}-2-\exp(-x/2)-\exp(-3x/2)\right) \\
        & \leq \left(\sqrt{2\pi}-2-\exp(-1/2)-\exp(-3/2)\right) < 0.
    \end{aligned}
\end{equation*}
Given $h_1(0) <0 $, it follows that $g(x) \leq \Phi(x)$ for $x \in [0,1]$.
For $x \in [1, \infty)$,  we aim to show $1-g(x) \geq \phi(x)/x \geq 1-\Phi(x)$ using the standard Gaussian distribution's tail bound.
To this end, we define $h_2(x) := \sqrt{2\pi}x\exp\left(\frac{x^2}{2}\right)-(1+\exp(x))$ and show $h_2(x) \geq 0$ for $x\in [1, \infty)$. This holds since $h_2(1) >0 $ and 
\[h_2'(x) = \sqrt{2\pi}(1+x^2)\exp\left(\frac{x^2}{2}\right) - \exp(x) \geq h_2'(1) >0 \text{ for } x \geq 1.\]
Together with
\[\frac{g'(x)}{1-g(x)} = g(x) \leq 1, \quad \lim_{x\to \infty} x^2g'(x) = \frac{x^2 e^x}{(1+\exp(x))^2} = 0, \]
\textit{Assumption 2.3(b)} is thereby confirmed. For \textit{Assumption 2.3(c)}, it is sufficient to show that 
\[\lim_{x\to \infty} |g''(x)| = 0.\]

\subsection{Probit link Function}
The probit link function is defined as 
\[g(x) = \Phi(x) := \frac{1}{\sqrt{2\pi}} \int_{-\infty}^{x} \exp\left\{-\frac{t^2}{2}\right\} dt,\]
and its derivatives are given by 
\[g'(x) = \frac{1}{\sqrt{2\pi}} \exp\left\{-\frac{x^2}{2}\right\}, \quad g''(x) = -\frac{x}{\sqrt{2\pi}} \exp\left\{-\frac{x^2}{2}\right\}.\]
\textit{Assumption 2.3(a)} holds since $0 < g'(x) \leq \frac{1}{\sqrt{2\pi}}$ for $x \in \mathbb{R}$ and $g''(x) <0 $ for $x>0$. It is trivial to show that $g(x) \leq \Phi(x)$ for all $x >0$.
We can show that
\[ \lim_{x \to \infty} \frac{g'(x)}{1-g(x)} = \lim_{x \to \infty} -\frac{1}{x} = 0, \quad  \lim_{x\to \infty} x^2g'(x) = \lim_{x\to \infty} \frac{x^2}{\sqrt{2\pi} \exp\left\{\frac{x^2}{2}\right\}} = 0,\]
and
\[\lim_{x \to \infty} |g''(x)| = \lim_{x \to -\infty} |g''(x)| = \lim_{x \to \infty} \frac{x}{\sqrt{2\pi}\exp\left\{\frac{x^2}{2}\right\}} = 0,\]
so \textit{Assumption 2.3(b)} and \textit{2.3(c)} are satisfied.

\subsection{Student's distribution}
The link function is defined by the student's $t_v$-distribution:
\[g(x) = \int_{-\infty}^{x} g'(t) dt,\]
and its derivatives are given by
\[g'(x) = \frac{\Gamma_d(\frac{\nu+1}{2})}{\sqrt{\nu \pi}\Gamma_d(\frac{\nu}{2})}\left(1+\frac{x^2}{\nu}\right)^{-(\nu+1)/2}, \quad 
g''(x) = \frac{\Gamma_d(\frac{\nu+1}{2})}{\sqrt{\nu \pi}\Gamma_d(\frac{\nu}{2})} \cdot \left(-\frac{\nu+1}{2}\right)\left(1+\frac{x^2}{\nu}\right)^{-(\nu+3)/2}.\]
where $\nu \in \mathbb{N}$ is the degree of freedom and $\Gamma_d$ is the gamma function. \textit{Assumption 2.3(a)} holds since $0 < g'(x) \leq \frac{1}{\sqrt{2\pi}}$ for $x \in \mathbb{R}$ and $g''(x) <0 $ for $x>0$. The fact that $g(x) \leq \Phi(x)$ for all $x >0$ holds by the definition of student's distribution.
We can show that
\[\lim_{x \to \infty} \frac{g'(x)}{1-g(x)} = \lim_{x \to \infty} \frac{\nu+1}{2}\left(1+\frac{x^2}{\nu}\right)^{-1} = 0,\]
\[\lim_{x\to \infty} x^2g'(x) = \frac{\sqrt{\nu \pi}\Gamma_d(\frac{\nu}{2})}{\Gamma_d(\frac{\nu+1}{2})} \cdot \lim_{x \to \infty} \frac{x^2}{\left(1+\frac{x^2}{\nu}\right)^{(\nu+1)/2}} = 0, \]
if $\nu > 1$ and $\lim_{x\to \infty} x^2g'(x)$ is bounded by a constant if $\nu = 1$, so \textit{Assumption 2.3(b)} is satisfied. \textit{Assumption 2.3(c)} holds by the fact that $\lim_{x \to \infty} |g''(x)| = 0$.

\section{ADDITIONAL SIMULATION RESULTS}
This appendix reports designs, implementation details, and full results that complement the one-page simulation report in the main text.

\subsection{Simulation Results of Section 4 under Additional Design Settings}\label{sec:mdesigns}
We consider a two-stage model, described as $Y_i \sim \text{Bernoulli}\left(g(X_i^\top \beta_Y + D_i)\right)$,
$D_i \sim \text{Bernoulli}\left(g(X_i^\top \beta_D)\right)$ where $g(x) = \exp(x)/(1+\exp(x))$. Covariates $X_i$ are generated from a multivariate Gaussian distribution with a covariance matrix $\Sigma$ where $\Sigma_{ij} = \rho^{|i-j|}$. 
Three designs of coefficients are examined.\footnote{Design (a) coincides with the main-text setting; we restate it here for completeness and notational consistency.}
\setcounter{bean}{0}%   ← template-style counter initialisation
\begin{list}{Design (\alph{bean})}{\usecounter{bean}}

  \item Sparse outcome; sparse \emph{vs.} dense propensity.\smallskip

        \begin{list}{(\roman{beana})}{\usecounter{beana}} % second level
          \item \emph{Outcome Coefficients: }  
                Let $(\beta_Y)_j\propto 1/j^{2}$ (sparse);

        \item \emph{Propensity Coefficients: } 
                We study two cases
                (1) sparse $(\beta_D)_j\propto 1/j^{2}$ and  
                (2) dense $(\beta_D)_j\propto 1/\sqrt{j}$,  
                each rescaled to $\|\beta_D\|_2\in\{1,4\}$.
        \end{list}

  \item Extreme-probability outcomes.\smallskip

        \begin{list}{(\roman{beana})}{\usecounter{beana}} % second level
          \item \emph{Outcome Coefficients: } Varying
        $\|\beta_Y\|_2$ across $\{10,20,\dots,80\}$ so that
        $g(X_i^{\top}\beta_Y)$ is concentrated near $0$ or $1$.
        \item \emph{Propensity Coefficients: } 
                Keep $(\beta_D)_j\propto 1/j^{2}$ with $\|\beta_D\|_2=1$.
        \end{list}

  \item Cross-sparsity mix.\smallskip

        \begin{list}{(\roman{beana})}{\usecounter{beana}}
          \item \emph{Outcome coefficients: }$\beta_Y$ is chosen from among four alternatives: (1) \textit{Dense}, defined as $(\beta_Y)j \propto 1/\sqrt{j}$; (2) \textit{Harmonic}, expressed as $(\beta_Y)j \propto 1/(j+9)$; (3) \textit{Moderately Sparse} and (4) \textit{Sparse} formulated as
\[\beta_Y \propto (\underbrace{10,...,10}_{10}, \underbrace{1,...,1}_{90}, \underbrace{0,...,0}_{p-100}), \quad \beta_Y \propto (\underbrace{1,...,1}_{10}, \underbrace{0,...,0}_{p-10}),\]
respectively.

          \item \emph{Propensity coefficients: }  
                $\beta_D$ is chosen from either of the two following possibilities: (1) \textit{Dense}, where $(\beta_D)j \propto 1/\sqrt{j}$; (2) \textit{Sparse}, represented by 
\[\beta_D \propto (\underbrace{1,...,1}_{10}, \underbrace{0,...,0}_{p-10}).\] 
        \end{list}

\end{list}
This arrangement expands the diversity of settings and allows for a comprehensive evaluation of our proposed method under various conditions.

We evaluate our estimator in three implementations: DB1 (2-fold sample splitting), DB2 (2-fold cross-fitting), and DB6 (no splitting).\footnote{We set the tuning parameter to $\zeta=0.5$ in the optimization program, following \citet{athey2018approximate}. At $\zeta=0.5$ the objective targets worst-case MSE (a standard minimax criterion; see \citet{armstrong2018optimal}, \citet{armstrong2021finite}). In practice, results are stable across a wide range of $\zeta$ (see Section \ref{sec:diffzeta}); developing a data-driven choice is left for future work.}
We benchmark against established methods.\footnote{We do not include double selection (e.g., \citet{belloni2016post}) because (i) DML by \citet{chernozhukov2018double} can be viewed as an extension of that procedure by 
    estimating both the outcome model and the propensity score model with flexible ML and is included here, and (ii) for linear outcomes, \citet{athey2018approximate} already compare to double selection.} These are: (i) Naïve difference in means $\bar Y_t-\bar Y_c$; (ii) Regression with a simple average of residuals;\footnote{We omit the pure regression plug-in, whose bias is documented in the Section 1, and instead include a plug-in plus simple residual averaging to reduce bias.} (iii) IPW with propensities from Lasso-logit, trimmed to $(0.05,0.95)$; (iv) DML with Lasso-logit outcome and propensity models, reported with two-fold sample splitting (DML1) and two-fold cross-fitting (DML2);\footnote{All outcome and propensity score models are estimated via Lasso logistic models, and for regularization parameter selection, we employed cross-validation with the \textit{lambda.min} rule using the R package \texttt{glmnet} (\citet{friedman2010regularization}).} (v) AML with two-fold cross-fitting;\footnote{Implemented from the code accompanying \citet{chernozhukov2022automatic} on the \textit{Econometrica} website.} and (vi) ARB by \citep{athey2018approximate}, which assumes a linear outcome model and uses no splitting. Although the GLM outcome is misspecified for ARB, the linear approximation may still perform well when signals are weak because our target averages predictions.

To assess the performance of the diverse estimators, we employ the (relative) MSE as the evaluative metric, defined as:
\begin{equation*}
    \frac{1}{n_{\text{sim}}}\sum_{i = 1}^{n_{\text{sim}}} \left(\frac{\hat{\tau}_i - \tau}{\tau}\right)^2,
\end{equation*}
where $n_{\text{sim}}$ denotes the number of replications.
\begin{table}
  \caption{Mean‐squared error of different estimators under Design (a) setting}
  \begin{center}
    \begin{tabular*}{\textwidth}{@{}lcccccccc@{}}
      \hline\hline
      \multirow{3}{*}{Method}
        & \multicolumn{4}{c}{Sparse $\beta_D$}
        & \multicolumn{4}{c}{Dense $\beta_D$} \\
      \cline{2-5} \cline{6-9}
        & \multicolumn{2}{c}{$\|\beta_D\|_2=1$}
        & \multicolumn{2}{c}{$4$}
        & \multicolumn{2}{c}{$\|\beta_D\|_2 = 1$}
        & \multicolumn{2}{c}{$4$} \\
      \cline{2-3}\cline{4-5}\cline{6-7}\cline{8-9}
        & $\|\beta_Y\|_2=1$ & $4$
        & $1$ & $4$
        & $\|\beta_Y\|_2=1$ & $4$
        & $1$ & $4$ \\
      \hline
      Naïve      & 1.176 & 4.565 & 2.214 & 8.785 & 0.430 & 1.539 & 0.646 & 2.338 \\
      Regression & 0.162 & 0.224 & 0.262 & 0.411 & 0.112 & 0.165 & 0.144 & 0.210 \\
      IPW        & 0.196 & 0.630 & 0.373 & 1.222 & 0.162 & 0.487 & 0.227 & 0.711 \\
      DML1       & 0.164 & 0.255 & 0.330 & 0.431 & 0.236 & 0.305 & 0.494 & 0.698 \\
      DML2       & 0.113 & 0.153 & 0.290 & 0.339 & 0.123 & 0.201 & 0.631 & 0.369 \\
      AML        & 0.115 & 0.169 & 0.170 & 0.361 & 0.098 & 0.156 & 0.112 & 0.168 \\
      ARB        & 0.094 & 0.141 & 0.210 & 0.315 & 0.095 & 0.158 & 0.113 & 0.191 \\
      DB1        & 0.126 & 0.162 & 0.202 & 0.291 & 0.125 & 0.186 & 0.146 & 0.204 \\
      DB2        & 0.087 & 0.123 & 0.144 & 0.221 & 0.098 & 0.146 & 0.110 & 0.164 \\
      DB6        & 0.067 & 0.089 & 0.116 & 0.161 & 0.070 & 0.099 & 0.078 & 0.108 \\
      \hline\hline
    \end{tabular*}
  \end{center}
  \footnotesize
  \renewcommand{\baselineskip}{11pt}
  \textbf{Note:}  In this study, we generate data under the prescribed settings: $n = 500$, $p = 800$, $\rho = 0.5$, and $n_{\text{sim}} = 1,000$. Denote
\textit{DML1} as the Double Machine Learning estimator employing two-fold sample-splitting, and \textit{DML2} as the same estimator
utilizing two-fold cross-fitting. Our proposed debiasing technique utilizing sample-splitting is termed \textit{DB1}. When augmented
with cross-fitting, it is labeled \textit{DB2}, and without any splitting techniques, it is referred to as \textit{DB6}.
\label{tab:RMSE_1}
\end{table}

\begin{table}
    \caption{Mean-squared error of different estimators under Design (b) setting}
    \label{tab:RMSE_2}
    \begin{center}
    \begin{tabular*}{\textwidth}{@{}lcccccccc@{}}
        \hline \hline
        \multirow{2}{*}{Method} & \multicolumn{8}{c}{$||\beta_Y||_2$} \\
        \cline{2-9}
        & $\quad 10 \quad $ & 
        $\quad 20 \quad $ & 
        $\quad 30 \quad $ & 
        $\quad 40 \quad $ & 
        $\quad 50 \quad $ & 
        $\quad 60 \quad $ &
        $\quad 70 \quad $ & 
        $\quad 80 \quad $ \\
        \hline
        Na\"{i}ve & 2.729 & 3.000 & 3.088 & 3.135 & 3.153 & 3.183 & 3.228 & 3.210 \\ 
        Regression & 0.286 & 0.297 & 0.307 & 0.295 & 0.300 & 0.300 & 0.306 & 0.301 \\ 
        IPW & 0.607 & 0.819 & 1.014 & 1.055 & 1.148 & 1.173 & 1.207 & 1.226 \\ 
        DML1 & 2.151 & 15.848 & 15.221 & 3.685 & 12.616 & 27.307 & 1.564 & 2.701 \\ 
        DML2 & 11.942 & 5.872 & 15.379 & 2.652 & 3.156 & 1.312 & 0.769 & 0.490 \\ 
        AML & 0.228 & 0.279 & 0.309 & 0.305 & 0.313 & 0.333 & 0.338 & 0.345 \\ 
        ARB & 0.354 & 0.524 & 0.716 & 0.752 & 0.847 & 0.884 & 0.937 & 0.987 \\ 
        DB1 & 0.275 & 0.336 & 0.371 & 0.362 & 0.412 & 0.401 & 0.422 & 0.417 \\ 
        DB2 & 0.196 & 0.238 & 0.273 & 0.269 & 0.300 & 0.300 & 0.298 & 0.314 \\ 
        DB3 & 0.156 & 0.191 & 0.227 & 0.221 & 0.249 & 0.250 & 0.260 & 0.258 \\ 
        \hline
        \hline
        \end{tabular*}
    \end{center}
  \footnotesize
  \renewcommand{\baselineskip}{11pt}
        \textbf{Note:} In this study, we generate data under the prescribed settings: $n = 500$, $p = 800$, $\rho = 0.5$, and $n_{\text{sim}} = 1,000$. 
         Signal strengths are standardized such that $||\beta_D||_2 = 1$ and we vary the magnitude of $||\beta_Y||_2$ across a range from 10 to 80.
         For additional clarifications, refer to the notes in Table \ref{tab:RMSE_1}.
\end{table}

\begin{table}
    \caption{Mean-squared error of different estimators under Design (c) setting}
    \label{tab:RMSE_3}
    \begin{center}
    \begin{tabular*}{\textwidth}{@{}lcccccccc@{}}
        \hline \hline
        \multirow{2}{*}{Method} & \multicolumn{2}{c}{Sp $\beta_Y$} & \multicolumn{2}{c}{Mod Sp $\beta_Y$} & \multicolumn{2}{c}{Har $\beta_Y$} & \multicolumn{2}{c}{Den $\beta_Y$} \\
        \cline{2-3} \cline{4-5} \cline{6-7} \cline{8-9}
        & Sp $\beta_D$ & Den $\beta_D$ & Sp $\beta_D$ & Den $\beta_D$ & Sp $\beta_D$ & Den $\beta_D$ & Sp $\beta_D$ & Den $\beta_D$ \\
        \hline
        Na\"{i}ve & 3.118 & 1.133 & 2.906 & 1.747 & 1.749 & 3.056 & 1.226 & 3.465 \\ 
        Regression & 0.515 & 0.224 & 0.507 & 0.476 & 0.493 & 1.310 & 0.474 & 1.915 \\ 
        IPW & 0.688 & 0.410 & 0.653 & 0.726 & 0.460 & 1.633 & 0.345 & 2.102 \\ 
        DML1 & 0.501 & 0.374 & 0.517 & 0.694 & 0.492 & 1.861 & 0.473 & 2.385 \\ 
        DML2 & 0.376 & 0.293 & 0.421 & 0.562 & 0.422 & 1.642 & 0.345 & 2.221 \\ 
        AML & 0.436 & 0.253 & 0.450 & 0.575 & 0.414 & 1.738 & 0.344 & 2.360 \\ 
        ARB & 0.294 & 0.179 & 0.311 & 0.344 & 0.276 & 1.071 & 0.235 & 1.436 \\ 
        DB1 & 0.383 & 0.264 & 0.409 & 0.526 & 0.400 & 1.536 & 0.345 & 2.092 \\ 
        DB2 & 0.334 & 0.221 & 0.343 & 0.476 & 0.335 & 1.455 & 0.288 & 2.022 \\ 
        DB3 & 0.177 & 0.123 & 0.180 & 0.271 & 0.184 & 0.863 & 0.167 & 1.315 \\   
        \hline
        \hline
    \end{tabular*}
    \end{center}
  \footnotesize
  \renewcommand{\baselineskip}{11pt}
  \textbf{Note:} In this study, we generate data under the prescribed settings: $n = 500$, $p = 800$, $\rho = 0.5$, and $n_{\text{sim}} = 1,000$. 
        Signal strengths are standardized such that $||\beta_D||_2 = 1$ and $||\beta_Y||_2 = 1$ for all models. 
        For additional clarifications, refer to the notes in Table \ref{tab:RMSE_1}.
\end{table}

In Design (a), the coefficients for the outcome model are sparse, whereas the propensity score model's coefficients are considered under both sparse and dense scenarios. 
Table \ref{tab:RMSE_1} reveals that methods relying on the inverse of the estimated propensity score, such as IPW, DML1, and DML2, do not outperform the regression-based imputation estimator with simple weighting,
under scenarios with dense $\beta_D$ or strong strength signals even for sparse $\beta_D$.
This finding aligns with the challenges associated with using inverse propensity scores, as discussed in Section 2. 
Specifically, when $\beta_D$ is dense, it becomes challenging to obtain accurate estimates of the propensity score function.
Our newly introduced debiasing estimator with sample-splitting, labeled as DB1, outperforms other methods, including Na\"{i}ve, Regression, IPW, DML1, DML2, and AML, in most cases, demonstrating stability. 
Moreover, when our proposed estimator with cross-fitting and without any splitting techniques, referred to as DB2 and DB6, it yields superior results, besting all other estimators in every case.
Additionally, the new version of double machine learning estimator that does not use the inverse propensity score weighting, namely AML, also outperforms DML1 and DML2, under dense $\beta_D$.
This reinforces our argument about the instability associated with the use of inverse estimated propensity scores.
Note that the linear misspecified version of our proposed estimator, denoted as ARB, performs well, particularly under weak signal strengths. 
Since ARB utilizes the entire sample without splitting, it is more appropriate to compare it with DB2 and DB6. Both of our proposed estimators outperform ARB across all scenarios.

In Design (b), while holding other parameters constant, we only adjust the signal strength of \(\beta_Y\) to ensure that the expected values of outcomes do not concentrate around the middle of the interval \([0,1]\). 
As the signal strength increases, shown in Table \ref{tab:RMSE_2}, the performance of the ARB estimator deteriorates because the linear approximation becomes less appropriate. 
Estimators that are based on the correct specification of the outcome regression, such as Regression, AML, and DB, perform better. 
However, it is important to note that the double machine learning estimators exhibit exceptionally poor performance in this setting. 
This could be due to both estimators not being stable, leading to extreme estimates that greatly increase the MSE. We report the median squared error of our simulations in Section \ref{sec:medianSE}.

In Design (c), various combinations of $\beta_D$ and $\beta_Y$ are examined with a fixed signal strength. 
The results in Table \ref{tab:RMSE_3} indicate that when $\beta_D$ is dense, DML2 fails to outperform the regression estimator, whereas DML2 excels when $\beta_D$ is sparse. 
This pattern emerges because a simpler propensity score model yields more accurate finite sample estimates of both the propensity score and its inverse. 
Conversely, complications arise as the propensity score model becomes more intricate. In this case, estimators that bypass estimating the propensity score, such as AML, DB1-DB3, shows better performance. Moreover, the estimator, DB3, consistently outperform other methods in all scenarios.

\subsection{Simulation Results of Section 4 in Median Squared Error Form}\label{sec:medianSE}
In this section, we present simulation results for the three settings described in (Section 4 and) Section \ref{sec:mdesigns} of the paper, focusing on the median squared error. Because the DML estimator exhibits instability, especially in the Design (b) setting.

\begin{table}
    \caption{Median squared error of different estimates under Design (a) setting}
    \begin{center}
    \begin{tabular*}{\textwidth}{@{}lcccccccc@{}}
      \hline\hline
      \multirow{3}{*}{Method}
        & \multicolumn{4}{c}{Sparse $\beta_D$}
        & \multicolumn{4}{c}{Dense $\beta_D$} \\
      \cline{2-5} \cline{6-9}
        & \multicolumn{2}{c}{$\|\beta_D\|_2=1$}
        & \multicolumn{2}{c}{$4$}
        & \multicolumn{2}{c}{$\|\beta_D\|_2 = 1$}
        & \multicolumn{2}{c}{$4$} \\
      \cline{2-3}\cline{4-5}\cline{6-7}\cline{8-9}
        & $\|\beta_Y\|_2=1$ & $4$
        & $1$ & $4$
        & $\|\beta_Y\|_2=1$ & $4$
        & $1$ & $4$ \\
        \hline
        Na\"{i}ve & 1.146 & 4.515 & 2.160 & 8.735 & 0.367 & 1.461 & 0.593 & 2.221 \\ 
        Regression & 0.124 & 0.163 & 0.215 & 0.338 & 0.064 & 0.106 & 0.104 & 0.149 \\ 
        IPW & 0.143 & 0.567 & 0.278 & 1.025 & 0.108 & 0.408 & 0.169 & 0.602 \\ 
        DML1 & 0.081 & 0.125 & 0.123 & 0.183 & 0.088 & 0.156 & 0.131 & 0.181 \\ 
        DML2 & 0.058 & 0.088 & 0.095 & 0.128 & 0.062 & 0.102 & 0.073 & 0.111 \\ 
        AML & 0.080 & 0.108 & 0.126 & 0.276 & 0.057 & 0.100 & 0.068 & 0.099 \\ 
        ARB & 0.041 & 0.066 & 0.105 & 0.144 & 0.041 & 0.074 & 0.053 & 0.100 \\ 
        DB1 & 0.062 & 0.074 & 0.095 & 0.152 & 0.057 & 0.093 & 0.075 & 0.108 \\ 
        DB2 & 0.043 & 0.065 & 0.078 & 0.125 & 0.047 & 0.083 & 0.063 & 0.100 \\ 
        DB6 & 0.029 & 0.046 & 0.055 & 0.080 & 0.032 & 0.053 & 0.036 & 0.060 \\ 
        \hline
        \hline
    \end{tabular*}
    \end{center}
    \footnotesize
  \renewcommand{\baselineskip}{11pt}
  \textbf{Note:} In this study, we generate data under the prescribed settings: $n = 500$, $p = 800$, $\rho = 0.5$, and $n_{\text{sim}} = 1,000$. 
        Denote \textit{DML1} as the Double Machine Learning estimator employing two-fold sample-splitting, and \textit{DML2} as the same estimator utilizing two-fold cross-fitting. 
        Our proposed debiasing technique utilizing sample-splitting is termed \(\textit{DB1}\). When augmented with cross-fitting, it is labelled \(\textit{DB2}\), and without any splitting techniques, it is referred to as \(\textit{DB6}\).
        \label{tab:median_RMSE_1}
\end{table}

\begin{table}
    \caption{Median squared error of different estimates under Design (b) setting}
    \label{tab:median_RMSE_2}
    \begin{center}
    \begin{tabular*}{\textwidth}{@{}lcccccccc@{}}
        \hline \hline
        \multirow{2}{*}{Method} & \multicolumn{8}{c}{$||\beta_Y||_2$} \\
        \cline{2-9}
        & $\quad 10 \quad $ & 
        $\quad 20 \quad $ & 
        $\quad 30 \quad $ & 
        $\quad 40 \quad $ & 
        $\quad 50 \quad $ & 
        $\quad 60 \quad $ &
        $\quad 70 \quad $ & 
        $\quad 80 \quad $ \\
        \hline
        Na\"{i}ve & 2.678 & 2.969 & 3.049 & 3.099 & 3.117 & 3.162 & 3.168 & 3.152 \\ 
        Regression & 0.242 & 0.250 & 0.269 & 0.242 & 0.251 & 0.255 & 0.261 & 0.270 \\ 
        IPW & 0.476 & 0.687 & 0.868 & 0.925 & 1.012 & 1.050 & 1.063 & 1.090 \\ 
        DML1 & 0.157 & 0.156 & 0.178 & 0.174 & 0.184 & 0.174 & 0.180 & 0.173 \\ 
        DML2 & 0.124 & 0.128 & 0.165 & 0.151 & 0.167 & 0.162 & 0.155 & 0.163 \\ 
        AML & 0.176 & 0.229 & 0.257 & 0.266 & 0.271 & 0.285 & 0.293 & 0.294 \\ 
        ARB & 0.195 & 0.310 & 0.397 & 0.458 & 0.521 & 0.597 & 0.576 & 0.676 \\ 
        DB1 & 0.123 & 0.151 & 0.196 & 0.175 & 0.208 & 0.196 & 0.236 & 0.206 \\ 
        DB2 & 0.098 & 0.125 & 0.158 & 0.163 & 0.196 & 0.190 & 0.187 & 0.195 \\ 
        DB6 & 0.078 & 0.090 & 0.105 & 0.103 & 0.122 & 0.131 & 0.144 & 0.140 \\ 
        \hline
        \hline
    \end{tabular*}
    \end{center}
    \footnotesize
  \renewcommand{\baselineskip}{11pt}
        \textbf{Note:} In this study, we generate data under the prescribed settings: $n = 500$, $p = 800$, $\rho = 0.5$, and $n_{\text{sim}} = 1,000$. 
         Signal strengths are standardized such that $||\beta_D||_2 = 1$ and we vary the magnitude of $||\beta_Y||_2$ across a range from 10 to 80.
         For additional clarifications, refer to the notes in Table \ref{tab:median_RMSE_1}.
\end{table}

\begin{table}
    \caption{Median squared error of different estimates under Design (c) setting}
    \label{tab:median_RMSE_3}
    \begin{center}
    \begin{tabular*}{\textwidth}{@{}lcccccccc@{}}
        \hline \hline
        \multirow{2}{*}{Method} & \multicolumn{2}{c}{Sp $\beta_Y$} & \multicolumn{2}{c}{Mod Sp $\beta_Y$} & \multicolumn{2}{c}{Har $\beta_Y$} & \multicolumn{2}{c}{Den $\beta_Y$} \\
        \cline{2-3} \cline{4-5} \cline{6-7} \cline{8-9}
        & Sp $\beta_D$ & Den $\beta_D$ & Sp $\beta_D$ & Den $\beta_D$ & Sp $\beta_D$ & Den $\beta_D$ & Sp $\beta_D$ & Den $\beta_D$ \\
        \hline
        Na\"{i}ve & 3.073 & 1.072 & 2.877 & 1.678 & 1.712 & 3.012 & 1.144 & 3.404 \\
        Regression & 0.449 & 0.166 & 0.445 & 0.404 & 0.424 & 1.214 & 0.386 & 1.805 \\ 
        IPW & 0.622 & 0.348 & 0.610 & 0.666 & 0.363 & 1.583 & 0.247 & 2.043 \\ 
        DML1 & 0.331 & 0.200 & 0.329 & 0.494 & 0.291 & 1.545 & 0.246 & 2.142 \\ 
        DML2 & 0.274 & 0.188 & 0.301 & 0.480 & 0.276 & 1.551 & 0.229 & 2.099 \\ 
        AML & 0.375 & 0.201 & 0.407 & 0.516 & 0.362 & 1.680 & 0.290 & 2.297 \\ 
        ARB & 0.182 & 0.097 & 0.204 & 0.251 & 0.159 & 0.980 & 0.134 & 1.359 \\ 
        DB1 & 0.270 & 0.172 & 0.284 & 0.411 & 0.254 & 1.419 & 0.201 & 1.955 \\ 
        DB2 & 0.247 & 0.157 & 0.270 & 0.419 & 0.246 & 1.380 & 0.203 & 1.913 \\ 
        DB6 & 0.112 & 0.073 & 0.117 & 0.216 & 0.110 & 0.790 & 0.093 & 1.195 \\   
        \hline
        \hline
    \end{tabular*}
    \end{center}
    \footnotesize
  \renewcommand{\baselineskip}{11pt}
  \textbf{Note:} In this study, we generate data under the prescribed settings: $n = 500$, $p = 800$, $\rho = 0.5$, and $n_{\text{sim}} = 1,000$. 
        Signal strengths are standardized such that $||\beta_D||_2 = 1$ and $||\beta_Y||_2 = 1$ for all models. 
        For additional clarifications, refer to the notes in Table \ref{tab:median_RMSE_1}.
\end{table}

Based on the results from Tables \ref{tab:median_RMSE_1} to \ref{tab:median_RMSE_3}, particularly Table \ref{tab:median_RMSE_2}, it is evident that the DML1 and DML2 estimators exhibit good performance in terms of median squared error. This observation with Table \ref{tab:RMSE_1} - \ref{tab:RMSE_3} suggests that while the DML method generally yields reliable results across most replications, it occasionally produces anomalously high estimates. Conversely, our proposed estimators perform well in both mean squared error and median squared error, indicative of their robustness and stability.

\subsection{Simulation Results Using More Folds on Cross-fitting}
In this section, we discuss the simulation outcomes for the three scenarios outlined in Section \ref{sec:mdesigns}, focusing on the performance of DML with an increased number of cross-fitting folds and our proposed estimator.

\begin{table}
    \caption{Mean-squared error for estimators with different splits under Design (a) setting}
    \label{tab:morefolds_RMSE_1}
    \begin{center}
    \begin{tabular*}{\textwidth}{@{}lcccccccc@{}}
      \hline\hline
      \multirow{3}{*}{Method}
        & \multicolumn{4}{c}{Sparse $\beta_D$}
        & \multicolumn{4}{c}{Dense $\beta_D$} \\
      \cline{2-5} \cline{6-9}
        & \multicolumn{2}{c}{$\|\beta_D\|_2=1$}
        & \multicolumn{2}{c}{$4$}
        & \multicolumn{2}{c}{$\|\beta_D\|_2 = 1$}
        & \multicolumn{2}{c}{$4$} \\
      \cline{2-3}\cline{4-5}\cline{6-7}\cline{8-9}
        & $\|\beta_Y\|_2=1$ & $4$
        & $1$ & $4$
        & $\|\beta_Y\|_2=1$ & $4$
        & $1$ & $4$ \\
        \hline
        DML1 & 0.169 & 0.245 & 1.550 & 0.532 & 0.201 & 0.336 & 0.454 & 0.850 \\ 
  DML2 & 0.104 & 0.162 & 0.179 & 0.237 & 0.156 & 0.199 & 0.399 & 0.314 \\ 
  DML3 & 0.083 & 0.127 & 0.209 & 0.182 & 0.100 & 0.154 & 0.193 & 0.396 \\ 
  DML4 & 0.073 & 0.115 & 0.125 & 0.167 & 0.096 & 0.141 & 0.234 & 0.244 \\ 
  DML5 & 0.075 & 0.112 & 0.139 & 0.151 & 0.086 & 0.136 & 0.226 & 0.387 \\ 
  DB1 & 0.114 & 0.173 & 0.193 & 0.266 & 0.118 & 0.176 & 0.152 & 0.230 \\ 
  DB2 & 0.083 & 0.128 & 0.132 & 0.202 & 0.092 & 0.135 & 0.108 & 0.177 \\ 
  DB3 & 0.083 & 0.124 & 0.133 & 0.202 & 0.094 & 0.141 & 0.115 & 0.185 \\ 
  DB4 & 0.085 & 0.134 & 0.144 & 0.214 & 0.103 & 0.155 & 0.123 & 0.199 \\ 
  DB5 & 0.091 & 0.141 & 0.155 & 0.221 & 0.109 & 0.172 & 0.132 & 0.219 \\ 
  DB6 & 0.065 & 0.094 & 0.107 & 0.147 & 0.067 & 0.092 & 0.078 & 0.115 \\ 
        \hline
        \hline
    \end{tabular*}
    \end{center}
    \footnotesize
  \renewcommand{\baselineskip}{11pt}
  \textbf{Note:} In this study, data are generated according to the defined parameters: sample size \( n = 500 \), number of predictors \( p = 800 \), correlation \( \rho = 0.5 \), and number of simulations \( n_{\text{sim}} = 1,000 \). We introduce several estimators: Double Machine Learning (DML) estimators are denoted as \(\textit{DML1}\), employing two-fold sample-splitting; \(\textit{DML2}\), using two-fold cross-fitting; and \(\textit{DML3}\), \(\textit{DML4}\), and \(\textit{DML5}\) for three-, four-, and five-fold cross-fitting, respectively. Our proposed debiasing technique, \(\textit{DB1}\), utilizes two-fold sample-splitting. This approach is extended in \(\textit{DB2}\) through \(\textit{DB5}\) to include two- to five-fold cross-fitting, respectively, while \(\textit{DB6}\) represents our method implemented without any splitting techniques.
\end{table}

\begin{table}
    \caption{Mean-squared error for estimators with different splits under Design (b) setting}
    \label{tab:morefolds_RMSE_2}
    \begin{center}
    \begin{tabular*}{\textwidth}{@{}lcccccccc@{}}
        \hline \hline
        \multirow{2}{*}{Method} & \multicolumn{8}{c}{$||\beta_Y||_2$} \\
        \cline{2-9}
        & $\quad 10 \quad $ & 
        $\quad 20 \quad $ & 
        $\quad 30 \quad $ & 
        $\quad 40 \quad $ & 
        $\quad 50 \quad $ & 
        $\quad 60 \quad $ &
        $\quad 70 \quad $ & 
        $\quad 80 \quad $ \\
        \hline
        DML1 & 428.684 & 16.163 & 41.339 & 2.611 & 1.078 & 5.376 & 2.560 & 0.763 \\ 
  DML2 & 0.963 & 4.942 & 129.389 & 1.894 & 2.354 & 24.046 & 1.071 & 0.830 \\ 
  DML3 & 2.239 & 8.006 & 44.118 & 0.764 & 1.996 & 0.714 & 0.653 & 5.191 \\ 
  DML4 & 0.313 & 0.404 & 0.770 & 1.654 & 0.716 & 1.472 & 0.736 & 0.441 \\ 
  DML5 & 0.887 & 0.620 & 0.910 & 0.795 & 2.031 & 3.520 & 0.466 & 1.211 \\ 
  DB1 & 0.274 & 0.327 & 0.366 & 0.402 & 0.402 & 0.421 & 0.467 & 0.442 \\ 
  DB2 & 0.195 & 0.243 & 0.256 & 0.295 & 0.290 & 0.296 & 0.326 & 0.317 \\ 
  DB3 & 0.209 & 0.267 & 0.300 & 0.337 & 0.337 & 0.349 & 0.379 & 0.367 \\ 
  DB4 & 0.227 & 0.301 & 0.342 & 0.366 & 0.382 & 0.391 & 0.421 & 0.411 \\ 
  DB5 & 0.241 & 0.332 & 0.380 & 0.410 & 0.418 & 0.425 & 0.462 & 0.461 \\ 
  DB6 & 0.154 & 0.202 & 0.217 & 0.240 & 0.239 & 0.241 & 0.268 & 0.260 \\ 
        \hline
        \hline
    \end{tabular*}
    \end{center}
    \footnotesize
  \renewcommand{\baselineskip}{11pt}
  \textbf{Note:} In this study, we generate data under the prescribed settings: $n = 500$, $p = 800$, $\rho = 0.5$, and $n_{\text{sim}} = 1,000$. 
         Signal strengths are standardized such that $||\beta_D||_2 = 1$ and we vary the magnitude of $||\beta_Y||_2$ across a range from 10 to 80.
         For additional clarifications, refer to the notes in Table \ref{tab:morefolds_RMSE_1}.
\end{table}

\begin{table}
    \caption{Mean-squared error for estimators with different splits under Design (c) setting}
    \label{tab:morefolds_RMSE_3}
    \begin{center}
    \begin{tabular*}{\textwidth}{@{}lcccccccc@{}}
        \hline \hline
        \multirow{2}{*}{Method} & \multicolumn{2}{c}{Sp $\beta_Y$} & \multicolumn{2}{c}{Mod Sp $\beta_Y$} & \multicolumn{2}{c}{Har $\beta_Y$} & \multicolumn{2}{c}{Den $\beta_Y$} \\
        \cline{2-3} \cline{4-5} \cline{6-7} \cline{8-9}
        & Sp $\beta_D$ & Den $\beta_D$ & Sp $\beta_D$ & Den $\beta_D$ & Sp $\beta_D$ & Den $\beta_D$ & Sp $\beta_D$ & Den $\beta_D$ \\
        \hline
        DML1 & 0.587 & 0.402 & 0.512 & 9.539 & 0.488 & 1.876 & 0.456 & 2.398 \\ 
  DML2 & 0.388 & 0.299 & 0.424 & 0.595 & 0.377 & 1.652 & 0.332 & 2.432 \\ 
  DML3 & 0.261 & 0.203 & 0.278 & 0.418 & 0.268 & 1.304 & 0.255 & 1.857 \\ 
  DML4 & 0.221 & 0.180 & 0.241 & 0.353 & 0.241 & 1.175 & 0.230 & 1.736 \\ 
  DML5 & 0.208 & 0.170 & 0.224 & 0.330 & 0.217 & 1.080 & 0.217 & 1.637 \\ 
  DB1 & 0.385 & 0.263 & 0.417 & 0.516 & 0.385 & 1.524 & 0.363 & 2.053 \\ 
  DB2 & 0.339 & 0.225 & 0.345 & 0.476 & 0.307 & 1.479 & 0.276 & 1.976 \\ 
  DB3 & 0.364 & 0.259 & 0.373 & 0.519 & 0.315 & 1.501 & 0.267 & 1.970 \\ 
  DB4 & 0.395 & 0.280 & 0.394 & 0.545 & 0.322 & 1.512 & 0.262 & 1.995 \\ 
  DB5 & 0.420 & 0.298 & 0.419 & 0.568 & 0.330 & 1.541 & 0.264 & 1.997 \\ 
  DB6 & 0.176 & 0.129 & 0.184 & 0.273 & 0.161 & 0.869 & 0.164 & 1.329 \\ 
        \hline
        \hline
    \end{tabular*}
    \end{center}
    \footnotesize
  \renewcommand{\baselineskip}{11pt}
  \textbf{Note:} In this study, we generate data under the prescribed settings: $n = 500$, $p = 800$, $\rho = 0.5$, and $n_{\text{sim}} = 1,000$. 
        Signal strengths are standardized such that $||\beta_D||_2 = 1$ and $||\beta_Y||_2 = 1$ for all models. 
        For additional clarifications, refer to the notes in Table \ref{tab:morefolds_RMSE_1}
\end{table}

Tables \ref{tab:morefolds_RMSE_1} to \ref{tab:morefolds_RMSE_3} reveal that increasing the number of folds in cross-fitting enhances both the performance and stability of the DML estimator, while adding more folds does not significantly benefit our proposed estimator. However, under Design (b) setting (Table \ref{tab:morefolds_RMSE_2}), even with four or five folds, the DML still exhibits some instability. Notably, our proposed method, which does not employ any splitting techniques, consistently outperforms the DML estimators across various settings.

\subsection{Simulation Results Using Different Values of the Tunning Parameter}\label{sec:diffzeta}
In this section, we report simulation results for the three settings in Section \ref{sec:mdesigns}, focusing on the performance of our proposed method without splitting (DB6).

\begin{table}
    \caption{Mean squared error of the proposed estimator with different tunning values under Design (a) setting}
    \label{tab:diffzeta_RMSE_1}
    \begin{center}
    \begin{tabular*}{\textwidth}{@{}lcccccccc@{}}
      \hline\hline
      \multirow{3}{*}{Method}
        & \multicolumn{4}{c}{Sparse $\beta_D$}
        & \multicolumn{4}{c}{Dense $\beta_D$} \\
      \cline{2-5} \cline{6-9}
        & \multicolumn{2}{c}{$\|\beta_D\|_2=1$}
        & \multicolumn{2}{c}{$4$}
        & \multicolumn{2}{c}{$\|\beta_D\|_2 = 1$}
        & \multicolumn{2}{c}{$4$} \\
      \cline{2-3}\cline{4-5}\cline{6-7}\cline{8-9}
        & $\|\beta_Y\|_2=1$ & $4$
        & $1$ & $4$
        & $\|\beta_Y\|_2=1$ & $4$
        & $1$ & $4$ \\
        \hline
        0.2 & 0.071 & 0.127 & 0.118 & 0.229 & 0.072 & 0.112 & 0.090 & 0.148 \\ 
  0.3 & 0.068 & 0.114 & 0.108 & 0.188 & 0.069 & 0.104 & 0.086 & 0.135 \\ 
  0.4 & 0.067 & 0.105 & 0.105 & 0.166 & 0.068 & 0.100 & 0.084 & 0.128 \\ 
  0.5 & 0.066 & 0.099 & 0.104 & 0.153 & 0.067 & 0.096 & 0.082 & 0.120 \\ 
  0.6 & 0.066 & 0.094 & 0.106 & 0.143 & 0.068 & 0.093 & 0.081 & 0.115 \\ 
  0.7 & 0.066 & 0.091 & 0.107 & 0.138 & 0.068 & 0.090 & 0.082 & 0.111 \\ 
  0.8 & 0.067 & 0.087 & 0.109 & 0.131 & 0.070 & 0.086 & 0.084 & 0.108 \\ 
        \hline
        \hline
    \end{tabular*}
    \end{center}
    \footnotesize
  \renewcommand{\baselineskip}{11pt}
  \textbf{Note:} In this study, we generate data under the prescribed settings: $n = 500$, $p = 800$, $\rho = 0.5$, and $n_{\text{sim}} = 1,000$. 
        We implement our proposed method without any splitting techniques by taking different values of the tuning parameter $\zeta$ in the optimization problem of Algorithm 1.
\end{table}

\begin{table}
    \caption{Mean squared error of the proposed estimator with different tunning values under Design (b) setting}
    \label{tab:diffzeta_RMSE_2}
    \begin{center}
    \begin{tabular*}{\textwidth}{@{}lcccccccc@{}}
        \hline \hline
        \multirow{2}{*}{Method} & \multicolumn{8}{c}{$||\beta_Y||_2$} \\
        \cline{2-9}
        & $\quad 10 \quad $ & 
        $\quad 20 \quad $ & 
        $\quad 30 \quad $ & 
        $\quad 40 \quad $ & 
        $\quad 50 \quad $ & 
        $\quad 60 \quad $ &
        $\quad 70 \quad $ & 
        $\quad 80 \quad $ \\
        \hline
        0.2 & 0.176 & 0.214 & 0.226 & 0.232 & 0.245 & 0.252 & 0.272 & 0.256 \\ 
  0.3 & 0.159 & 0.204 & 0.220 & 0.220 & 0.238 & 0.242 & 0.261 & 0.250 \\ 
  0.4 & 0.154 & 0.199 & 0.215 & 0.223 & 0.241 & 0.243 & 0.259 & 0.254 \\ 
  0.5 & 0.151 & 0.200 & 0.222 & 0.228 & 0.248 & 0.245 & 0.263 & 0.260 \\ 
  0.6 & 0.151 & 0.204 & 0.228 & 0.240 & 0.262 & 0.254 & 0.277 & 0.268 \\ 
  0.7 & 0.155 & 0.211 & 0.241 & 0.254 & 0.271 & 0.268 & 0.295 & 0.282 \\ 
  0.8 & 0.156 & 0.226 & 0.253 & 0.271 & 0.291 & 0.291 & 0.317 & 0.297 \\
        \hline
        \hline
    \end{tabular*}
    \end{center}
    \footnotesize
  \renewcommand{\baselineskip}{11pt}
  \textbf{Note:} In this study, we generate data under the prescribed settings: $n = 500$, $p = 800$, $\rho = 0.5$, and $n_{\text{sim}} = 1,000$. 
         Signal strengths are standardized such that $||\beta_D||_2 = 1$ and we vary the magnitude of $||\beta_Y||_2$ across a range from 10 to 80.
         For additional clarifications, refer to the notes in Table \ref{tab:diffzeta_RMSE_1}.
\end{table}

\begin{table}
    \caption{Mean squared error of the proposed estimator with different tunning values under Design (b) setting}
    \label{tab:diffzeta_RMSE_3}
    \begin{center}
    \begin{tabular*}{\textwidth}{@{}lcccccccc@{}}
        \hline \hline
        \multirow{2}{*}{Method} & \multicolumn{2}{c}{Sp $\beta_Y$} & \multicolumn{2}{c}{Mod Sp $\beta_Y$} & \multicolumn{2}{c}{Har $\beta_Y$} & \multicolumn{2}{c}{Den $\beta_Y$} \\
        \cline{2-3} \cline{4-5} \cline{6-7} \cline{8-9}
        & Sp $\beta_D$ & Den $\beta_D$ & Sp $\beta_D$ & Den $\beta_D$ & Sp $\beta_D$ & Den $\beta_D$ & Sp $\beta_D$ & Den $\beta_D$ \\
        \hline
        0.2 & 0.271 & 0.146 & 0.274 & 0.368 & 0.222 & 1.057 & 0.212 & 1.624 \\ 
  0.3 & 0.226 & 0.134 & 0.228 & 0.338 & 0.190 & 0.994 & 0.182 & 1.530 \\ 
  0.4 & 0.201 & 0.125 & 0.202 & 0.317 & 0.172 & 0.926 & 0.170 & 1.446 \\ 
  0.5 & 0.184 & 0.118 & 0.188 & 0.297 & 0.163 & 0.864 & 0.161 & 1.363 \\ 
  0.6 & 0.175 & 0.113 & 0.174 & 0.278 & 0.155 & 0.809 & 0.156 & 1.273 \\ 
  0.7 & 0.164 & 0.109 & 0.165 & 0.261 & 0.150 & 0.753 & 0.150 & 1.189 \\ 
  0.8 & 0.156 & 0.105 & 0.157 & 0.245 & 0.147 & 0.693 & 0.148 & 1.098 \\ 
        \hline
        \hline
    \end{tabular*}
    \end{center}
    \footnotesize
  \renewcommand{\baselineskip}{11pt}
  \textbf{Note:} In this study, we generate data under the prescribed settings: $n = 500$, $p = 800$, $\rho = 0.5$, and $n_{\text{sim}} = 1,000$. 
        Signal strengths are standardized such that $||\beta_D||_2 = 1$ and $||\beta_Y||_2 = 1$ for all models. 
        For additional clarifications, refer to the notes in Table \ref{tab:diffzeta_RMSE_1}.
\end{table}

Tables \ref{tab:diffzeta_RMSE_1}–\ref{tab:diffzeta_RMSE_3} demonstrate that, in most scenarios, increasing the tuning parameter 
$\zeta$ yields modest improvements in our estimator’s performance by tightening control over the worst-case bias (maximal covariate imbalance). This aligns with the intuition that, when the variance function is well behaved and does not inflate, it is advantageous to place greater emphasis on bias control by increasing $\zeta$. 
Crucially, across all values of 
$\zeta$, our proposed estimator is stable and maintains strong performance, outperforming the alternative methods.

\subsection{Simulation Results for the Confidence Intervals}
We assess the coverage of confidence intervals for an array of $n$ and $p$ values under the Design (a) setting of Section \ref{sec:mdesigns}. The confidence interval constructed based on Theorem 3.2:
\[(\hat{\tau} - 1.96\sqrt{\hat{v}_\tau},\ \hat{\tau} + 1.96\sqrt{\hat{v}_\tau}).\]
\begin{sidewaystable}
    \caption{Coverage of the confidence interval under Design (a) setting}
    \label{tab:CI_1}
    \begin{center}
    \begin{tabular*}{\textwidth}{@{}cccccccccccccc@{}}
        \hline \hline
        $n$ & $p$ & \multicolumn{6}{c}{Coverage rate (\%)} & \multicolumn{6}{c}{Length} \\
        \cline{3-8}  \cline{9-14}
        & & \multicolumn{3}{c}{Sparse $\beta_D$} & \multicolumn{3}{c}{Dense $\beta_D$} & \multicolumn{3}{c}{Sparse $\beta_D$} & \multicolumn{3}{c}{Dense $\beta_D$} \\
        \cline{3-5} \cline{6-8} \cline{9-11} \cline{12-14}

        & & DB1 & DB2 & DB6 & DB1 & DB2 & DB6 & DB1 & DB2 & DB6 & DB1 & DB2 & DB6 \\
        \hline
        500 & 600 & 0.933 & 0.932 & 0.939 & 0.911 & 0.913 & 0.931 & 0.253 & 0.252 & 0.201 & 0.230 & 0.230 & 0.185 \\
        1,000 & 600 & 0.944 & 0.944 & 0.953 & 0.932 & 0.919 & 0.945 & 0.182 & 0.182 & 0.144 & 0.163 & 0.163 & 0.132 \\
        500 & 1,200 & 0.928 & 0.929 & 0.945 & 0.912 & 0.910 & 0.932 & 0.252 & 0.252 & 0.200 & 0.229 & 0.229 & 0.183 \\
        1,000 & 1,200 & 0.942 & 0.942 & 0.953 & 0.916 & 0.926 & 0.936 & 0.181 & 0.181 & 0.144 & 0.162 & 0.162 & 0.131 \\
         \hline
        \hline
    \end{tabular*}
    \end{center}
  \footnotesize
  \renewcommand{\baselineskip}{11pt}
  \textbf{Note:} In this study, we generate data under the prescribed settings:  $\rho = 0.5$ and $n_{\text{sim}} = 2,000$. 
        Signal strengths are standardized such that $||\beta_D||_2 = 1$ and $||\beta_Y||_2 = 1$ for all models. The target coverage is 0.95.
        For additional clarifications, refer to the notes in Table \ref{tab:RMSE_1}.
\end{sidewaystable}

According to Table \ref{tab:CI_1}, the coverage rate tends to improve as the sample size $n$ grows. 
When the propensity score model is comparatively simple, our proposed methods exhibit stability and closely align with the target coverage rate.
Additionally, DB6, which does not utilize any splitting techniques, demonstrates greater efficiency concerning the length of the confidence interval compared to DB1 and DB2. 
Therefore, in practical applications, we recommend employing our proposed method without splitting.
\end{document}